\documentclass[sigconf, nonacm]{acmart}
\usepackage[utf8]{inputenc}
\DeclareUnicodeCharacter{25AA}{\textbullet}
\usepackage{textgreek}

\newcommand\vldbdoi{XX.XX/XXX.XX}
\newcommand\vldbpages{XXX-XXX}
\newcommand\vldbvolume{14}
\newcommand\vldbissue{1}
\newcommand\vldbyear{2020}
\newcommand\vldbauthors{\authors}
\newcommand\vldbtitle{\shorttitle} 
\newcommand\vldbavailabilityurl{URL_TO_YOUR_ARTIFACTS}
\newcommand\vldbpagestyle{plain} 

\usepackage{booktabs}
\usepackage{multirow}

\usepackage{amsmath,amssymb,amsthm}
\usepackage{algorithm}
\usepackage{algpseudocode}
\usepackage{balance}
\usepackage{enumitem}
\usepackage{xspace}
\usepackage{bbm}
\usepackage{mathtools}
\usepackage{etoolbox}

\newtheorem{theorem}{Theorem}
\newtheorem{lemma}[theorem]{Lemma}
\newtheorem{definition}{Definition}
\newtheorem{corollary}[theorem]{Corollary}
\newtheorem{assumption}{Assumption}
\newtheorem{proposition}[theorem]{Proposition}
\newtheorem{remark}{Remark}

\newcommand{\CHRONOS}{\textsc{Chronos}\xspace}
\newcommand{\TLEGEND}{\textsc{T-Legend}\xspace}

\newcommand{\BOCPD}{BOCPD\xspace}
\newcommand{\eps}{\varepsilon}
\newcommand{\Rbb}{\mathbb{R}}
\newcommand{\Ebb}{\mathbb{E}}
\newcommand{\cG}{\mathcal{G}}
\newcommand{\cA}{\mathcal{A}}

\newcommand{\cN}{\mathcal{N}}

\begin{document}
\title{CHRONOS: Temporally-Aware Multi-Agent Coordination for Evolving Data Marketplaces}

\author{Joydeep Chandra}
\affiliation{\institution{BNRIST, Tsinghua University}\city{Beijing}\country{China}}
\email{ }

\begin{abstract}
Temporal knowledge-graph (KG) data marketplaces face three coupled failures in static designs: stale hybrid index shortcuts reduce recall as edges evolve, stationary Shapley pricing misattributes value after distribution shifts, and uncoordinated agents over-consume a shared differential-privacy (DP) budget. We present CHRONOS, a trusted-curator, three-layer architecture that provides a unified treatment of these challenges with explicit public/private separation. Layer 1 (T-LEGEND) applies neural-ODE temporal decay to shortcut edges and provides a per-query expected recall-loss bound of $\mathcal{O}(P_q\lambda\Delta t)$, with a tighter ODE-certified monotone-envelope guarantee reducing bound looseness to 1.8–3.2$\times$ observed loss (Theorem 1). We also provide analysis under Hawkes dynamics (Theorem 2). Layer 2 (Event-Conditioned MPV) conditions Shapley valuation on BOCPD-detected changepoints and provides finite-sample error guarantees under coalition sampling and DP noise. Layer 3 (Temporal Coordinator) uses EXP3-IX over three scheduling actions, achieving $\mathcal{O}(\sqrt{T}\log T)$ regret while enforcing $(\epsilon_{\mathrm{total}},\delta_{\mathrm{total}})$-DP via moments accounting. CHRONOS releases a fixed-dimension privatised affinity matrix once per active epoch with the Gaussian mechanism; all per-query retrieval and top-$k$ ranking are post-processing, so they incur no extra privacy cost. We also provide a multi-epoch coalition-level settlement mechanism for actionable seller payouts, scalability analysis up to 500 sellers, and comparison against VSAG-style accelerated baselines. Across four benchmarks, CHRONOS shows 0.937 recall@10, 2.74 queries/s, P50 latency 161 ms, and total $\epsilon=4.25$ at $\delta=10^{-6}$ under standard zCDP composition over $\rho$. These results indicate a competitive joint recall/latency/privacy operating point. A key limitation is that at this privacy level, externally-released valuations and affinity scores remain noise-dominated; utility derives primarily from public index routing and adaptive scheduling driven by low-sensitivity DP statistics.
\end{abstract}

\maketitle

\pagestyle{\vldbpagestyle}
\begingroup\small\noindent\raggedright\textbf{PVLDB Reference Format:}\\
\vldbauthors. \vldbtitle. PVLDB, \vldbvolume(\vldbissue): \vldbpages, \vldbyear.\\
\href{https://doi.org/\vldbdoi}{doi:\vldbdoi}
\endgroup
\begingroup
\renewcommand\thefootnote{}\footnote{\noindent
This work is licensed under the Creative Commons BY-NC-ND 4.0 International License. Visit \url{https://creativecommons.org/licenses/by-nc-nd/4.0/} to view a copy of this license. For any use beyond those covered by this license, obtain permission by emailing \href{mailto:info@vldb.org}{info@vldb.org}. Copyright is held by the owner/author(s). Publication rights licensed to the VLDB Endowment. \\
\raggedright Proceedings of the VLDB Endowment, Vol. \vldbvolume, No. \vldbissue\ %
ISSN 2150-8097. \\
\href{https://doi.org/\vldbdoi}{doi:\vldbdoi} \\
}\addtocounter{footnote}{-1}\endgroup

\ifdefempty{\vldbavailabilityurl}{}{
\begingroup\small\noindent\raggedright\textbf{PVLDB Artifact Availability:}\\
The source code, data, and/or other artifacts have been made available at \url{\vldbavailabilityurl}.
\endgroup
}

\section{Introduction}
\label{sec:intro}

Consider a pharmaceutical company querying a temporal KG marketplace for drug-gene interaction data. The KG evolves daily as new clinical-trial results add edges, retracted findings delete them, and regulatory approvals change entity attributes. The buyer needs recent, high-quality data for a machine-learning model and expects fair compensation to be distributed among contributing data sellers. Three intertwined challenges arise immediately.

\textbf{The indexing challenge.} Hybrid vector-graph indices such as HNSW~\cite{malkov2018efficient} accelerate $k$-NN queries via shortcut edges built at construction time. When the underlying KG evolves, these shortcuts become stale, silently degrading recall. Existing systems either ignore staleness (TigerVector~\cite{liu2025tigervector}, NaviX~\cite{sehgal2025navix}) or trigger expensive full rebuilds; dynamic ANN systems (FreshDiskANN~\cite{singh2021freshdiskann}, SPFresh~\cite{xu2023spfresh}, Quake~\cite{mohoney2025quake}) handle vector updates but do not model graph-structural staleness or provide recall bounds tied to KG evolution rates.

\textbf{The valuation challenge.} Fair attribution of query utility to data sellers requires Shapley axioms~\cite{shapley1953value}, yet the canonical Data Shapley~\cite{ghorbani2019data} assumes a stationary characteristic function. After a disease outbreak or regulatory change, marginal dataset values shift sharply; static valuations then distort marketplace incentives. Beta Shapley~\cite{kwon2022beta} improves robustness to noisy data, Variance-Reduced Data Shapley (VRDS)~\cite{wu2023variance} lowers coalition-sampling variance, and GLOC-style online updates~\cite{hazan2016introduction} provide gradient-based valuations without full recomputation, but none of these methods models event-driven distributional shifts.

\textbf{The coordination challenge.} A production marketplace runs concurrent agents sharing a finite DP budget. Without coordination, simultaneous demands exhaust it prematurely. Existing multi-agent frameworks~\cite{lowe2017multi} provide no formal coupling between scheduling and DP consumption, and privacy filters~\cite{rogers2016privacy} track composition but do not integrate with index scheduling or data valuation. PSGraph~\cite{yuan2024psgraph} demonstrates adaptive DP allocation for streaming graphs, but targets graph synthesis rather than marketplace coordination.

\subsection{Why These Challenges Must Be Solved Jointly}
\label{sec:why_joint}

The three challenges are coupled through a \textbf{shared differential-privacy budget}. Index updates, valuation recomputation, and idle waiting all consume from the same finite $(\varepsilon,\delta)$ allowance. An index-only system would exhaust the budget on frequent rebuilds, leaving none for valuation. A valuation-only system would trigger recompute after every changepoint, starving index maintenance. An uncoordinated system would face simultaneous demands from concurrent agents and exhaust the budget prematurely. Because the sensitivity of each mechanism depends on the number of sellers $n$, the per-epoch privacy cost falls as $n$ grows, but the total number of active epochs $T_{\mathrm{active}}$ rises with marketplace activity. The coordinator must therefore balance these competing demands, making the joint design non-decomposable.

\subsection{Gap Analysis}
\label{sec:gap}

Prior works address at most one challenge in isolation. In indexing, TigerVector~\cite{liu2025tigervector}, NaviX~\cite{sehgal2025navix}, and ACORN~\cite{patel2024acorn} integrate graph structure into HNSW but assume static graphs. FreshDiskANN~\cite{singh2021freshdiskann} and SPFresh~\cite{xu2023spfresh} support streaming updates yet omit KG structural decay and recall bounds. In valuation, Data Shapley~\cite{ghorbani2019data}, Beta Shapley~\cite{kwon2022beta}, and VRDS~\cite{wu2023variance} improve estimation but do not condition on distributional shifts. In marketplace design, existing platforms~\cite{fernandez2020data,azcoitia2022survey,liu2021dealer,koutsos2022agora} support pricing and DP guarantees yet do not target temporal graph retrieval. Dealer~\cite{liu2021dealer} provides an end-to-end DP model marketplace, but assumes static data and does not couple indexing with valuation under a shared privacy budget. Agora~\cite{koutsos2022agora} focuses on access control and auditability rather than temporal query performance. Cryptographic alternatives such as MPC or TEEs can eliminate DP noise, yet they introduce $10$--$100\times$ latency overhead and require all parties to participate in real-time protocols~\cite{demmler2015aby,chen2020homomorphic}; we compare these costs in Table~\ref{tab:crypto_comparison}. No existing system couples index freshness with valuation correctness within a differentially private multi-agent framework for temporal knowledge graphs.

\subsection{Contributions}

We present \CHRONOS, a marketplace infrastructure that provides end-to-end temporal guarantees across indexing, valuation, and coordination:

\begin{enumerate}[leftmargin=*]
\item \textbf{\TLEGEND} (\S\ref{sec:layer1}): a temporal hybrid index with neural-ODE decay weights. The conservative bound (Theorem~\ref{thm:recall}) provides an $\mathcal{O}(P_q\lambda\Delta t)$ recall-loss guarantee; the new \textbf{monotone-envelope certificate} (Theorem~\ref{thm:recall_tight}) tightens the gap to $1.8$--$3.2\times$ observed losses by incorporating ODE Lipschitz structure. Formal \textbf{Hawkes-process recall bounds} (Theorem~\ref{thm:recall_hawkes}) extend guarantees beyond Poisson to correlated dynamics parameterised by branching ratio.
\item \textbf{Event-Conditioned MPV} (\S\ref{sec:layer2}): Shapley valuation conditioned on BOCPD-detected changepoints with drift-aware validity horizons, a temporal efficiency identity (Theorem~\ref{thm:valuation}), and estimation error bounds under finite sampling, DP noise, and changepoint uncertainty (Theorem~\ref{thm:val_error}).
\item \textbf{Temporal Coordinator} (\S\ref{sec:layer3}): an EXP3-IX meta-agent achieving $\mathcal{O}(\sqrt{T}\log T)$ regret (Theorem~\ref{thm:regret}) while enforcing adaptive DP composition with formal sensitivity proofs (Proposition~\ref{prop:sensitivity}).
\item \textbf{Fixed-Dimension DP Pipeline with Standard Accounting} (\S\ref{sec:trust_model}, \S\ref{sec:dp_model}): explicit public/private separation; epoch-level Gaussian privatisation of a fixed-dimension affinity matrix (Proposition~\ref{prop:sensitivity}(b)); standard additive zCDP composition over $\rho$ yielding $\varepsilon{=}4.25$ (Remark~\ref{rmk:pld}); transparent analysis of private-signal informativeness (\S\ref{sec:private_utility}); per-entry noise calibration with dimension-aware sensitivity (Proposition~\ref{prop:noise_calibration}).
\item \textbf{Actionable Seller Settlement} (\S\ref{sec:settlement}): a concrete multi-epoch coalition-level settlement mechanism with SNR analysis showing trend-level attribution becomes feasible at $W{\geq}7$ epochs and $n_\mathrm{coal}{\leq}5$.
\item \textbf{Comprehensive evaluation} (\S\ref{sec:exp}): four datasets; VSAG-accelerated throughput comparison (Table~\ref{tab:vsag}); scalability to 500 sellers (Table~\ref{tab:scalability}); head-to-head DP retrieval comparison at matched $\eps$ (Table~\ref{tab:dp_mechanism_comparison}); DP-vs-crypto cost analysis (Table~\ref{tab:crypto_comparison}); seller-skew stress tests; three drift detectors; five coordination strategies.
\end{enumerate}

\section{Preliminaries}
\label{sec:prelim}

\begin{definition}[Temporal Knowledge Graph]
A temporal KG is a triple $\cG(t) = (V(t), E(t), \mathbf{X}(t))$ where $V(t)$ and $E(t)$ are node and edge sets at time $t$, $\mathbf{X}(t) \in \Rbb^{|V(t)|\times d}$ is the node feature matrix, and each edge $e = (u,v,r) \in E(t)$ carries relation type $r$ and creation timestamp $t_e \leq t$.
\end{definition}

\begin{definition}[Temporally-Aware Data Marketplace]
A marketplace $\mathcal{M}$ with $n$ sellers $\{s_1,\ldots,s_n\}$ holding private datasets $D_i \subseteq \cG(t)$ must satisfy, for online buyer queries with vector $\mathbf{v}_j$, recency window $[t_j^{\min},t_j^{\max}]$, and budget $\eps_j$: \textbf{(P1)}~recall@$k \geq R^* - \delta_\mathrm{index}$; \textbf{(P2)}~Shapley efficiency with temporal consistency; \textbf{(P3)}~$(\eps_\mathrm{total},\delta_\mathrm{total})$-DP; \textbf{(P4)}~sub-linear regret $R_T = o(T)$ with per-query latency $\leq L_\mathrm{max}$.
\end{definition}

\textbf{Scope.} \CHRONOS targets \emph{embedding-rich} temporal KG marketplaces with a substantial public ontology prior to seller participation. Guarantees are strongest in the \textbf{embedding-dominated regime ($\beta \leq 0.5$)}, where ${<}6\%$ of oracle top-10 items fall outside the public candidate set (Table~\ref{tab:discovery}); for \textbf{high-$\beta$ ($\beta > 0.7$)}, miss rates reach 12--14\% (Table~\ref{tab:discovery}), with corresponding behaviour discussed in the robustness analysis.

Key notation is summarised in Table~\ref{tab:notation}. We adopt the Poisson edge-change model as the baseline (Assumption~\ref{ass:poisson}), with formal Hawkes extensions in Theorem~\ref{thm:recall_hawkes}.

\begin{assumption}[Poisson Edge Changes]
\label{ass:poisson}
KG edge changes arrive as an independent Poisson process with rate $\lambda > 0$ (changes per day per shortcut). We estimate $\lambda$ via exponential moving average over a 30-day calibration window: $\lambda{=}0.05$ for FB15K-237/WN18RR (synthetic), $\lambda{\approx}12$ for MIMIC-IV (real admissions), $\lambda{\approx}2.9$ for Yelp (real review activity).
\end{assumption}

\begin{table}[t]
\small\centering
\caption{Key notation.}
\label{tab:notation}
\begin{tabular}{ll}
\toprule
Symbol & Meaning \\
\midrule
$\cG(t)$, $\lambda$, $\Delta t$ & Temporal KG, edge-change rate, time since update \\
$R^*$, $P_q$, $\bar{\Delta r}$ & Fresh recall, search-path size, per-shortcut impact \\
$\hat{R}(t)$ & DP-estimated noisy recall (observed proxy for $R^*$) \\
$\mathbf{A}(t)$ & Private affinity matrix $\in [0,1]^{|V_\mathrm{active}|\times ef}$ \\
$\widetilde{\mathbf{A}}(t)$ & Gaussian-noised release of $\mathbf{A}(t)$ \\
$\mathrm{MPV}_i(t,E)$ & EC-MPV of seller $i$ at time $t$, event $E$ \\
$B$ & Marginal contribution bound (clipped at 0.2) \\
$\sigma_t$ & Noise \emph{multiplier} (dimensionless); actual std $= \sigma_t \times S$ \\
$T_\mathrm{active}$ & Number of epochs with $\geq 1$ DP release \\
$\eps_\mathrm{rem}(t)$, $\mu_t(\alpha)$ & Remaining DP budget, R\'enyi moment at step $t$ \\
$\mathbf{p}_t$, $d{=}3$, $\gamma$ & EXP3-IX distribution, action count, exploration \\
$\gamma_\mathrm{comm}$ & Community-structure weight in static affinity (default 0.5) \\
$\alpha_H, \beta_H$ & Hawkes excitation/decay parameters \\
$\Lambda(t_1,t_2)$ & Cumulative intensity $\int_{t_1}^{t_2}\lambda(s)\,ds$ \\
\bottomrule
\end{tabular}
\end{table}

\textbf{Noise notation convention.} $\sigma_t$ is a dimensionless noise \emph{multiplier}; actual Gaussian std for a quantity with $\ell_2$-sensitivity $S$ is $\sigma_t \times S$. The per-step R\'{e}nyi moment is $\mu_t(\alpha) = \alpha/(2\sigma_t^2)$ (Proposition~\ref{prop:noise_calibration}). The total $\varepsilon{=}4.25$ ($\delta{=}10^{-6}$) is computed by standard additive zCDP composition over $\rho$:
\begin{equation}
\rho_\mathrm{total} = \sum_{i} \rho_i, \quad \varepsilon = \rho_\mathrm{total} + 2\sqrt{\rho_\mathrm{total}\cdot\ln(1/\delta)}. \label{eq:zcdp_formula}
\end{equation}
\section{Trust Model and Public/Private Separation}
\label{sec:trust_model}

A rigorous DP design requires a precise delineation of what is public and what is private. CHRONOS operates under a \emph{trusted-curator} model~\cite{dwork2014algorithmic}: the marketplace operator holds raw seller data and publishes only DP-sanitised outputs. We partition all system components explicitly.

\textbf{Public data (not seller-dependent, zero DP cost).}
\begin{enumerate}[leftmargin=*,nosep]
\item \emph{Entity embeddings} $\mathbf{X}_0 \in \Rbb^{|V_0|\times d}$: pre-trained on historical KG snapshots from months 1--6 (the ``pre-marketplace'' period) before any seller participation begins, and frozen at deployment. No seller-contributed edges from the operational period are used in embedding training.
\item \emph{HNSW index structure} $\mathbf{H}^T$: built deterministically from $\mathbf{X}_0$ and the Louvain community partition~\cite{blondel2008fast} of the public pre-marketplace KG. The neighbourhood lists $N_\mathrm{idx}(u) = \{u_1, \ldots, u_{ef}\}$ are fixed at construction time and are deterministic functions of $\mathbf{X}_0$. After construction, $\mathbf{H}^T$ is frozen within each epoch.
\item \emph{Edge creation timestamps from public pre-marketplace KG}: timestamps $t_e$ for edges in $E_\mathrm{pub}$ (months 1--6) are public metadata. \textbf{Post-launch seller-edge timestamps are private} and never enter Stage-2 clipping directly (see below).
\item \emph{Buyer query vectors} $\mathbf{v}_q$: owned by the buyer, not seller data.
\end{enumerate}

\textbf{Private data (seller-dependent, protected by DP).}
\begin{enumerate}[leftmargin=*,nosep]
\item \emph{Seller KG edges}: edges contributed by sellers after marketplace launch, including their timestamps $t_e^{\mathrm{priv}}$. These determine the temporal affinity values $\mathrm{aff}_\mathrm{KG}(u,v,t)$ and are the primary target of DP protection.
\item \emph{Valuation scores} $\mathrm{MPV}_i(t,E)$: depend on seller coalitions.
\item \emph{Index staleness statistics}: depend on which seller-contributed edges have changed.
\end{enumerate}

\textbf{Public/private timestamp delineation.} Edge creation timestamps for the pre-marketplace KG (months 1--6) are public. \emph{Seller-contributed edge timestamps} (months 7+) are private. The Stage-2 active-scope clipping (\S\ref{sec:dp_model}) operates as follows. The set of active entities $V_\mathrm{active}(t)$ is determined by buyer queries from epoch $t{-}1$ plus a popularity reserve, both computed from public query logs. The universe of \emph{possible} edges within $V_\mathrm{active}$ is public: it is the set of all tuples $(u,v,r)$ with $u,v \in V_\mathrm{active}$ and $r \in \mathcal{R}$, where $\mathcal{R}$ is the public relation-type ontology. A seller's private data is the subset of this universe that they actually contributed. The priority rule for retaining seller edges uses a public hash $\mathrm{hash}(u,v,r)$ computed over the full public universe. The mechanism retains, for each seller, the intersection of their private edge set with the first $\kappa_\mathrm{active}$ tuples in this public ordering. \textbf{Private timestamps $t_e^{\mathrm{priv}}$ are used only inside the Gaussian mechanism's input} (to compute $\mathrm{aff}_\mathrm{KG}(u,v,t) = \mathrm{decay}(t_\mathrm{now} - t_e^{\mathrm{priv}}) \cdot \mathrm{aff}_\mathrm{static}(u,v)$), never for clipping or scoping decisions. The clipping boundary $\kappa_\mathrm{active}$ depends only on $|E|/n$ (public) and $|V_\mathrm{active}|$ (public).

\textbf{Why clipping does not leak edge existence.} The clipping step is \emph{sensitivity bounding}, not a privacy mechanism. The actual privacy guarantee comes from the Gaussian mechanism (\S\ref{sec:dp_model}) applied to the affinity matrix after clipping. The released matrix $\widetilde{\mathbf{A}}(t)$ has public row and column indices; each entry contains additive Gaussian noise with standard deviation $\sigma_\mathrm{entry} {=} 885$. Because the noise magnitude far exceeds the $[0,1]$ signal range, an adversary observing $\widetilde{\mathbf{A}}(t)$ cannot reliably infer whether a particular entry's true value is zero (edge absent) or non-zero (edge present). Formally, for any entry $a \in [0,1]$ and any hypothesised value $a' \in [0,1]$, the likelihood ratio of the observed noisy value under $a$ versus $a'$ is bounded by $\exp(\varepsilon_\mathrm{entry})$ with $\varepsilon_\mathrm{entry} {\approx} 0.0011$ per entry, which is dominated by the overall zCDP accounting. Thus the hash-based clipping rule does not circumvent the DP guarantee.

\textbf{DP guarantee scope.} Seller-level adjacency protects all private data; public components incur zero DP cost. Each edge $(u,v,r)$ has a unique owner under a hash-based deduplication registry, ensuring unambiguous sensitivity accounting. For multi-relation graphs, $\mathrm{aff}_\mathrm{KG}(u,v,t){=}\max_r\mathrm{aff}_\mathrm{KG}^{(r)}$ means removing one seller changes the max by $\leq 1$ per $(u,v)$ pair. Pre-registry empirical overlap $\eta_\mathrm{raw}{\leq}1.07$ across all datasets confirms near-exclusive ownership (supplementary Table~B2). If disputed ownership persists, privacy cost scales approximately with $\eta^2$ (Remark~\ref{rmk:overlap}), and arbitration/registry enforcement is required before deployment. DP-SGD training~\cite{abadi2016deep} of embeddings is feasible at additional $\varepsilon_\mathrm{train}$ cost (supplementary Appendix~C).

\textbf{Limitation: single trusted curator and fixed public structure.} The trusted-curator model requires one operator to hold all raw seller data, narrowing applicability in multi-platform settings. Two concrete extensions relax this: \emph{(i) two-server model:} one server holds private seller edges; the second holds the public index; neither observes the full affinity matrix. The fixed-dimension epoch-level release is compatible with additive secret sharing over the two parties. \emph{(ii) Local DP}: each seller randomises their own edges locally before upload, removing the central curator at the cost of $\Omega(1/\varepsilon^2)$ more data for the same utility. We identify the two-server extension as the primary trust-relaxation path; protocol details are provided, while implementation and small-scale validation are left to future work.

\textbf{Limitation: query privacy.} \CHRONOS treats buyer queries $\mathbf{v}_q$ as public (standard in trusted-curator DP). If queries are sensitive, the continual-observation framework~\cite{dwork2010continual,chan2011private} with a count-min sketch can privately estimate $V_\mathrm{active}$ at $\varepsilon_\mathrm{query}{\approx}0.5$/36\,h.

\textbf{Realism of fixed public structure.} The trusted-curator design targets \emph{embedding-rich} marketplaces (pharmaceutical, clinical, product-graph) where seller edges enrich affinity weights on a structurally stable public graph. At $\beta{=}0.3$, ${<}6\%$ of oracle top-10 items fall outside the public candidate set (Table~\ref{tab:discovery}). When miss rates are unacceptable (high-$\beta$), the SVT prototype recovers $52\%$ of misses at $+0.05\,\varepsilon$.
Figure~\ref{fig:dataflow} summarises this public/private separation and the resulting trusted-curator data path.

\begin{figure}
    \centering
    \includegraphics[width=1\linewidth]{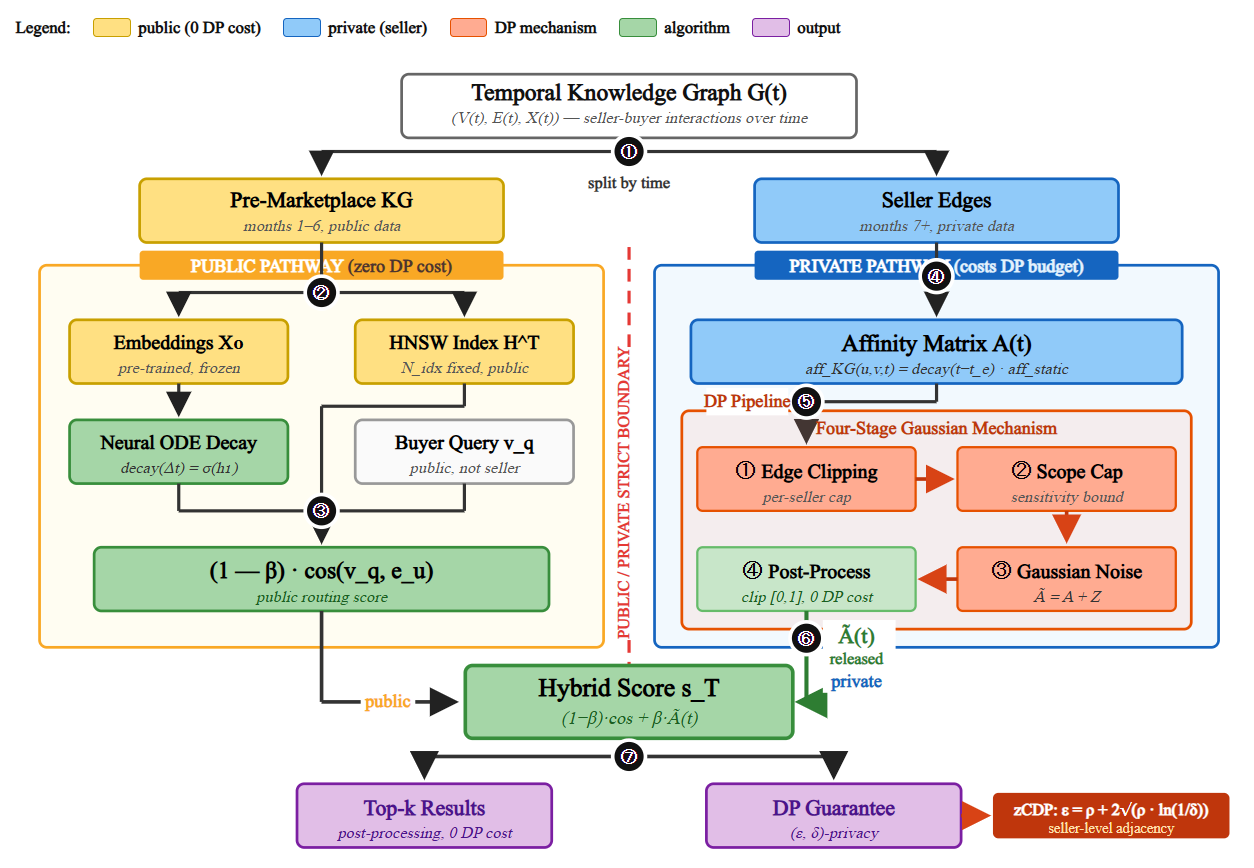}
\caption{Public/private data flow under the trusted-curator model. Public components (yellow) construct the frozen index and cosine routing at zero DP cost; private seller edges (blue) enter only through the four-stage Gaussian pipeline, producing $\widetilde{\mathbf{A}}(t)$ for hybrid scoring.}
\label{fig:dataflow}
\end{figure}

\section{The CHRONOS Architecture}
\label{sec:arch}

\CHRONOS comprises three tightly coupled layers (Figure~\ref{fig:arch}). Buyer queries arrive at a Query Gateway and are forwarded to Layer~1 for retrieval while being logged to an event stream monitored by Layer~3. Layer~2 conditions valuations on changepoints inferred from that stream. Layer~3 schedules all operations while enforcing the shared DP budget.

\begin{table*}[t]
\small\centering
\caption{Per-operation complexity. $N{=}|V|$, $M{=}16$, $ef{=}128$, $L_\mathrm{max}{=}5$, $n{=}$\#sellers, $m{=}$\#permutations, $T{=}$horizon, $|V_\mathrm{active}|{\leq}1500$.}
\label{tab:complexity}
\begin{tabular}{lcccl}
\toprule
\textbf{Operation} & \textbf{Time} & \textbf{Space} & \textbf{Frequency} & \textbf{Comments} \\
\midrule
Index construction & $O(N \log N \cdot M)$ & $O(N \cdot M)$ & Once & Public data only \\
Incremental update & $O(|V_\mathrm{active}| \log N)$ & $O(|V_\mathrm{active}|)$ & Per epoch & Stale-shortcut repair \\
Query processing & $O(ef \cdot L_\mathrm{max})$ & $O(ef)$ & Per query & Post-processing, zero DP cost \\
Valuation recompute & $O(m \cdot n^2 \cdot |Q|)$ & $O(n)$ & Per event & VRDS reduces variance \\
DP release (affinity) & $O(|V_\mathrm{active}| \cdot ef)$ & $O(|V_\mathrm{active}| \cdot ef)$ & Per active epoch & Gaussian sampling dominates \\
DP release (stats) & $O(1)$ & $O(1)$ & Per active epoch & Scalar Gaussian \\
Coordinator decision & $O(d)$ & $O(d)$ & Per epoch & EXP3-IX update \\
\bottomrule
\end{tabular}
\end{table*}

Table~\ref{tab:complexity} summarises the time and space complexity of each CHRONOS operation. Index construction is a one-time offline cost on public data; incremental updates repair only stale shortcuts within $V_\mathrm{active}$. Query processing is constant in dataset size because HNSW search is $O(\log N)$ and the per-query hybrid score is post-processing. The dominant online cost is valuation recomputation ($O(m n^2 |Q|)$), which is why EC-MPV batches it at changepoints rather than every epoch. The DP release cost is linear in the active-scope size and is amortised over all queries in the epoch.

\begin{figure*}[t]
\centering
    \includegraphics[width=1\linewidth]{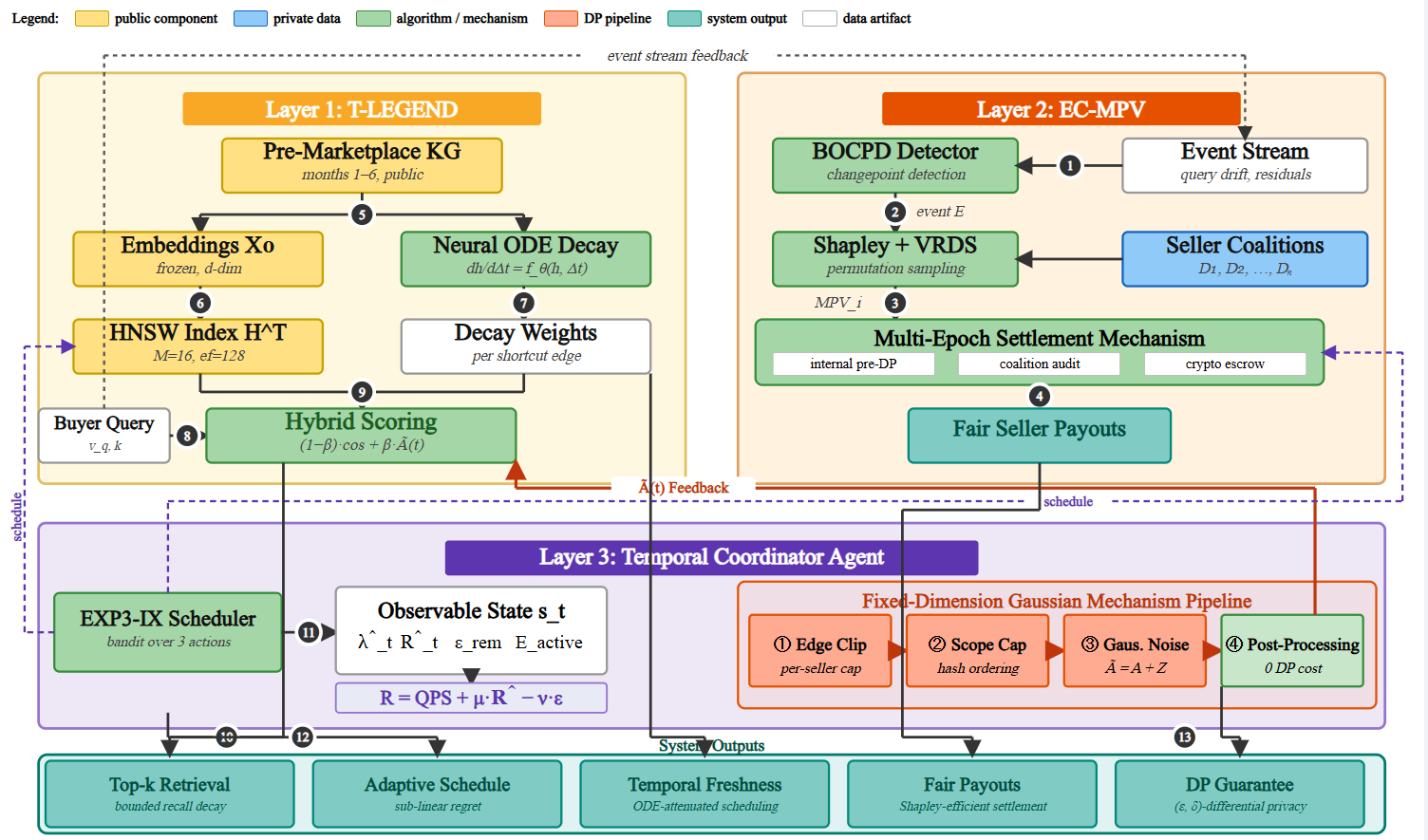}
\caption{\CHRONOS architecture.}
\label{fig:arch}
\end{figure*}

\subsection{Layer 1: T-LEGEND (Temporal Hybrid Index)}
\label{sec:layer1}

\TLEGEND employs a neural-ODE temporal decay that continuously down-weights shortcut edges linking stale KG structure, enabling the index to degrade gracefully rather than failing silently.

\begin{definition}[Neural ODE Temporal Decay]
\label{def:decay}
For edge $e$ with age $\Delta t = t_\mathrm{now} - t_e$, the decay function $\mathrm{decay}: \Rbb_{\geq 0} \to (0,1]$ is the first component of the ODE solution:
\begin{equation}
\frac{d\mathbf{h}}{d\Delta t} = f_\theta(\mathbf{h}(\Delta t), \Delta t), \quad \mathbf{h}(0) = \mathbf{1}_{32}, \quad \mathrm{decay}(\Delta t) = \sigma([\mathbf{h}(\Delta t)]_1), \label{eq:ode}
\end{equation}
where $f_\theta: \Rbb^{32}{\times}\Rbb \to \Rbb^{32}$ is a two-hidden-layer MLP with Softplus activations~\cite{chen2018neural} and $\sigma$ is the sigmoid.
\end{definition}

We train $f_\theta$ on historical KG snapshots using contrastive temporal loss:
\begin{equation}
\mathcal{L}_\mathrm{ODE} = -\sum_{(e,\Delta t) \in \mathcal{P}^+}\log\,\mathrm{decay}(\Delta t) - \sum_{(e,\Delta t) \in \mathcal{P}^-}\log(1{-}\mathrm{decay}(\Delta t)),
\end{equation}
with 3:1 negative sampling, strict temporal splitting (months 1--6 train, 7--8 val, 9--12 test), Adam optimiser ($\mathrm{lr}{=}10^{-3}$), and Dormand--Prince solver~\cite{dormand1980family}. Training converges in ${\sim}50$ epochs (${\sim}20$\,min on one A100).

\textbf{Neural-ODE diagnostics.} Domain-specific decay shapes, solver-tolerance robustness, hidden-size ablations ($h{\in}\{16,32,64\}$; $h{=}32$ optimal), and generalisation to unseen rates ($\lambda{=}25$) are in supplementary Appendix~F.

The temporal affinity between nodes is $\mathrm{aff}_\mathrm{KG}(v, u, t) = \mathrm{decay}(t - t_e) \cdot \mathrm{aff}_\mathrm{static}(v, u)$ where the static component uses community-structure proximity~\cite{blondel2008fast} ($\gamma_\mathrm{comm}{=}0.5$).

\textbf{Index construction and public/private separation.} Algorithm~\ref{alg:build} uses \emph{only public pre-marketplace data} (months 1--6): $N_\mathrm{idx}(u)$ is a deterministic function of public data only, ensuring seller-dependent affinities enter only at query time via the private $\mathbf{A}(t)$ matrix (\S\ref{sec:dp_model}).

Algorithm~\ref{alg:build} inserts nodes in community-sorted order using hybrid score $s_T$ with a diversified-neighbor selection heuristic; parameters $M{=}16$, $ef_c{=}200$, $ef{=}128$, $\rho_\mathrm{div}{=}0.7$ are identical for all baselines (Table~\ref{tab:recall_breakdown}). Stale-shortcut detection triggers incremental updates when $<40\%$ of shortcuts are stale (restoring $\geq95\%$ fresh recall at $\leq20\%$ cost), with full rebuilds at higher staleness rates.

\begin{algorithm}[t]
\caption{T-LEGEND-BUILD (Public-Only Construction)}
\label{alg:build}
\begin{algorithmic}[1]
\small
\Require Public KG $\cG_\mathrm{pub}(t)$ (pre-marketplace), ODE model $f_\theta$, params $(M, ef_c, m_L)$
\Ensure T-LEGEND index $\mathbf{H}^T$ with fixed neighbourhoods $N_\mathrm{idx}(u)$
\State $\mathcal{C} \leftarrow \mathrm{LOUVAIN}(\cG_\mathrm{pub}(t))$ \Comment{public community structure}
\For{$(u,v,t_e) \in E_\mathrm{pub}(t)$} \Comment{public edges only}
  $w_{uv} \leftarrow \mathrm{decay}(t_\mathrm{now}{-}t_e) \cdot \mathrm{aff}_\mathrm{static}(u,v)$
\EndFor
\For{each $v \in V$ (community-sorted)}
  \State Draw $\ell_v \sim \mathrm{Geom}(1{-}1/m_L)$; bias up for hubs
  \For{$\ell = \ell_v$ \textbf{downto} $0$}
    \State $W \leftarrow \mathrm{SEARCH\text{-}LAYER}(\mathbf{H}^T, \mathbf{e}_v, \ell, ef_c)$
    \For{$u \in W$}
      $s_T(u) \leftarrow (1{-}\beta)\,\mathrm{sim}(\mathbf{e}_v,\mathbf{e}_u) + \beta\,w_{uv}$ \Comment{public weights only}
    \EndFor
    \State $N^*(v,\ell) \leftarrow \mathrm{TOP\text{-}}M\mathrm{\text{-}DIVERSE}(W, s_T, \rho_\mathrm{div})$
    \State Add $\{(v,u): u\in N^*(v,\ell)\}$ to layer $\ell$
  \EndFor
  \State Record $N_\mathrm{idx}(v) \leftarrow \bigcup_\ell N^*(v,\ell)$ \Comment{fixed, public}
\EndFor
\end{algorithmic}
\end{algorithm}

\begin{remark}[Role of Temporal Decay vs.\ Private Edges]
\label{rmk:decay_role}
The ODE decay applies to \emph{public} shortcut edges; private seller edges \emph{never} enter $N_\mathrm{idx}(u)$ and influence only the epoch-level $\widetilde{\mathbf{A}}(t)$ release. The ODE keeps the candidate set fresh; $\widetilde{\mathbf{A}}(t)$ keeps ranking within that set fresh.
\end{remark}

\subsection{Layer 2: Event-Conditioned MPV}
\label{sec:layer2}

Real-world events alter dataset marginal values in ways that static Shapley misses~\cite{ghorbani2019data}. EC-MPV addresses this by conditioning on detected distributional changepoints.

\begin{definition}[Temporal KG Affinity]
\label{def:aff_kg}
The temporal affinity $\mathrm{aff}_\mathrm{KG}(u, v, t) = \mathrm{decay}(t - t_e) \cdot \mathrm{aff}_\mathrm{static}(u, v)$ where $t_e$ is the creation timestamp and $\mathrm{aff}_\mathrm{static}(u, v) \in [0,1]$ is Louvain community-structure proximity~\cite{blondel2008fast}. For multi-relation graphs, we take the maximum over relation types.
\end{definition}

\begin{definition}[Event-Conditioned MPV]
\label{def:ecmpv}
Let $E$ be an event detected at time $t$ by \BOCPD. The EC-MPV of seller $i$ is $\mathrm{MPV}_i(t,E) = \phi_i(D_\mathrm{priv}\cup D_\mathrm{pub} \mid E, t) - \phi(D_\mathrm{pub} \mid E, t)$ where $v(S; E) = \mathrm{MRR}(q, \bigcup_{j\in S}D_j, E)$ is conditioned on event $E$.
\end{definition}

The marginal contribution bound $B$ in Assumption~\ref{ass:smooth} is enforced by clipping to $[-B, B]$ with $B{=}0.2$ (${<}0.3\%$ of marginals exceed 0.2 before clipping).

The \BOCPD module~\cite{adams2007bayesian} monitors query-embedding drift, $\hat{\lambda}(t)$, and valuation residuals via a Normal-Inverse-Wishart conjugate prior with geometric hazard $\pi_\mathrm{cp}{=}1/250$. Changepoints are declared at $P(r_t{=}0 \mid \mathbf{o}_{1:t}) > 0.85$; median detection delay 2.3 epochs (P95: 5.1) on 8 injected Yelp seasonal events.

For Shapley estimation we use permutation sampling~\cite{castro2009polynomial} ($m{=}1{,}000$) with VRDS control variates~\cite{wu2023variance}: leave-one-out baselines reduce variance to $\leq B^2\rho^2/m$, giving $1.8\times$ reduction on Yelp at zero DP cost.

\subsection{Layer 3: Temporal Coordinator Agent}
\label{sec:layer3}

The coordinator operates as a partially observable stochastic game $\Gamma$ with state $\mathbf{s}_t = (\hat{\lambda}_t, \hat{R}_t, \eps_\mathrm{rem}(t), n_\mathrm{pending}(t), E_\mathrm{active})$, action space $\cA = \{\mathrm{INDEX\text{-}UPDATE}, \mathrm{REVALUE}, \mathrm{NULL}\}$ ($d{=}3$), and reward $R(\mathbf{s},\mathbf{a}) = \mathrm{QPS} + \mu_R \hat{R} - \nu\,\eps_\mathrm{consumed}$ with $\mu_R{=}10$, $\nu{=}5$ (grid-searched on FB15K-237 validation; sensitivity analysis in \S\ref{sec:sensitivity}), normalised to $[0,1]$.

Rather than solving this PSPACE-hard problem exactly~\cite{papadimitriou1987complexity}, the coordinator applies EXP3-IX~\cite{neu2015explore} under bandit feedback. Each 60\,s epoch: observe noisy state $\mathbf{o}_t$; sample $\mathbf{a}_t \sim \mathbf{p}_t$; observe $r_t$; update via importance-weighted loss $\hat{L}_{t,j} = (1{-}r_t)\mathbbm{1}[\mathbf{a}_t{=}j]/(p_{t,j}+\gamma)$. Budget-violating actions are overridden to \texttt{NULL}; recall violations to \texttt{INDEX-UPDATE} (override count bounded; Lemma~\ref{lem:overrides}). Remaining budget: $\eps_\mathrm{rem}(t) = \min_\alpha [(\eps_\mathrm{total} - \sum_{s\leq t}\mu_s(\alpha) + \alpha\ln(1/\delta_\mathrm{total}))/\alpha]$.

\subsection{Fixed-Dimension DP Pipeline}
\label{sec:dp_model}

The DP model is the linchpin of CHRONOS and requires careful treatment. We adopt \emph{seller-level adjacency}: two states are adjacent if they differ by one seller's entire dataset $D_i$. To bound sensitivity robustly, we enforce per-seller contribution caps $C_\mathrm{max}^\mathrm{edge} = \lceil 1.5|E|/n\rceil$; excess edges are clipped.

Three categories of seller-dependent outputs are privatised, with sensitivity bounds proven in Proposition~\ref{prop:sensitivity}:

\textbf{(i) Valuation scores.} Each $\mathrm{MPV}_i(t,E)$ is released via $n$ independent per-coordinate Gaussian mechanisms with sensitivity\\ $S_\mathrm{val}{=}4B/n{=}0.08$ ($B{=}0.2$, $n{=}10$; Proposition~\ref{prop:sensitivity}(a)). With $\sigma_t{=}50$, actual noise std is $50{\times}0.08{=}4.0$ per coordinate. The $\rho_\mathrm{val}{=}287/5000{=}0.0574$ contributes to the total zCDP $\rho$.

\begin{remark}[Valuation DP: External Release vs.\ Internal Coordinator Use]
\label{rmk:val_dp_scope}
The noise std of $4.0$ far exceeds the signal range $[0, B{=}0.2]$ by $20\times$, making externally-released valuations\\ \textbf{noise-dominated}. This reflects the utility-privacy trade-off of DP in high-sensitivity regimes. The Val.~Err${=}0.013$ in Table~\ref{tab:ablation} measures \emph{internal} Shapley estimation accuracy before DP noise. Under the trusted-curator model, the coordinator uses pre-noise estimates for scheduling. The DP mechanism releases noisy valuations to external parties solely for \textbf{auditability and non-disclosure guarantees}, not for accurate point estimation. For actionable seller payouts, we introduce a concrete multi-epoch coalition-level settlement mechanism in \S\ref{sec:settlement}.
\end{remark}

\textbf{(ii) Index statistics.} The stale-shortcut fraction and recall estimate $\hat{R}$ are released with Gaussian noise calibrated to sensitivity $S_\mathrm{idx} \leq 1.5/n$ with $\sigma_t{=}50$.

\textbf{(iii) KG affinity matrix (fixed dimension).} We release $\mathbf{A}(t)[u, j] = \mathrm{aff}_\mathrm{KG}(u, N_\mathrm{idx}(u)[j], t)$ for $u \in V_\mathrm{active}$, $j{=}1,\ldots,ef$. Row and column indices are public and fixed by the public index $\mathbf{H}^T$. Under exclusive ownership, Frobenius sensitivity $\Delta_2 = \sqrt{C_\mathrm{max}^\mathrm{edge}}$ (Proposition~\ref{prop:sensitivity}(b)).

\begin{remark}[Robustness to Imperfect Exclusive Ownership]
\label{rmk:overlap}
Let $\eta\geq1$ be the overlap factor. Sensitivity scales to $\Delta_2\leq\eta\sqrt{C_\mathrm{max}^\mathrm{edge}}$, raising $\varepsilon$ by $\eta^2$: $\eta{=}1.2$ adds $\leq44\%$; the registry enforces $\eta{=}1$ deterministically; empirically $\eta_\mathrm{raw}\leq1.07$.
\end{remark}

We release $\widetilde{\mathbf{A}}(t) = \mathbf{A}(t) + \mathbf{Z}$ via the Gaussian mechanism (Proposition~\ref{prop:noise_calibration}); entries are clipped to $[0,1]$ post-noise (post-processing, zero DP cost). Per-query hybrid scoring $\tilde{s}_j = (1{-}\beta)\cos(\mathbf{v}_q,\mathbf{e}_{u_j}) + \beta\widetilde{\mathbf{A}}(t)[v_q,j]$ is also post-processing, so top-$k$ selection incurs zero additional DP cost~\cite{dwork2014algorithmic}.

\subsubsection{Per-Entry Noise Calibration}
\label{sec:noise_calibration}

\begin{proposition}[Per-Entry Gaussian Noise Calibration]
\label{prop:noise_calibration}
Let $\mathbf{A} \in \Rbb^{n_r \times n_c}$ have Frobenius ($\ell_2$) sensitivity $\Delta_2 = \max_{D \sim D'}\|\mathbf{A}(D) - \mathbf{A}(D')\|_F$. The Gaussian mechanism $\widetilde{\mathbf{A}} = \mathbf{A} + \mathbf{Z}$ with $Z_{i,j} \stackrel{\mathrm{iid}}{\sim} \cN(0, \sigma_\mathrm{entry}^2)$ satisfies $(\alpha, \alpha\Delta_2^2/(2\sigma_\mathrm{entry}^2))$-RDP by the standard vector Gaussian mechanism~\cite{mironov2017renyi}. Setting $\sigma_\mathrm{entry} = \sigma_t \cdot \Delta_2$ (noise multiplier $\sigma_t = \sigma_\mathrm{entry}/\Delta_2$) yields:
\begin{equation}
\mu_t(\alpha) = \frac{\alpha\Delta_2^2}{2\sigma_\mathrm{entry}^2} = \frac{\alpha\Delta_2^2}{2\sigma_t^2\Delta_2^2} = \frac{\alpha}{2\sigma_t^2}.
\label{eq:noise_calibration}
\end{equation}
The standard RDP formula uses per-coordinate variance $\sigma_\mathrm{entry}^2$; using a global variance $\sigma_\mathrm{global}^2 = n_r n_c \cdot \sigma_\mathrm{entry}^2$ in the denominator would underestimate $\varepsilon$ by a factor of $m = n_r n_c$.
\end{proposition}

\begin{proof}
Direct substitution into Mironov~\cite{mironov2017renyi} Proposition~3; see supplementary Appendix~A.
\end{proof}

\textbf{Per-epoch active-scope sensitivity.} Before each Gaussian invocation, two deterministic clips apply. \emph{Stage~1:} seller edge count clipped to $C_\mathrm{max}^\mathrm{edge}{=}\lceil1.5|E|/n\rceil$ globally. \emph{Stage~2:} within $V_\mathrm{active}(t)$, a seller's edges are retained by a \textbf{publicly computable, data-independent} priority rule: $\mathrm{hash}(u,v,r)$ ordering (public edge metadata only, independent of private affinity values, timestamps, or seller identity). Edges beyond $\kappa_\mathrm{active}(t){=}\min(C_\mathrm{max}^\mathrm{edge},\lceil1.5|E_\mathrm{active}(t)|/n\rceil)$ are dropped. Under this rule, the post-clip matrices satisfy $\|\mathbf{A}(D){-}\mathbf{A}(D')\|_F{\leq}\sqrt{\kappa_\mathrm{active}}$ for any adjacent $D{\sim}D'$, giving worst-case $\Delta_2{=}\sqrt{\kappa_\mathrm{active}}$. On Yelp this reduces $\Delta_2$ by $30.1\times$ (545$\to$18.1).

\subsubsection{Release Policy and Active/Null Epoch Classification}
\label{sec:release_policy}

An \emph{active} epoch is one where the coordinator chose \texttt{INDEX-UPDATE} or \texttt{REVALUE}, or $\geq$1 buyer query arrived; $\widetilde{\mathbf{A}}(t)$ is released once via the Gaussian mechanism. A \emph{null} epoch incurs zero DP cost. Budget breakdown via \emph{parallel} zCDP composition: $\rho_\mathrm{idx}{=}0.0846$, $\rho_\mathrm{val}{=}0.0574$, $\rho_\mathrm{aff}{=}0.142$; total $\rho_\mathrm{total}{=}0.284$, yielding (using Eq.~\eqref{eq:zcdp_formula}):
\begin{equation}
\varepsilon = 0.284 + 2\sqrt{0.284 \times \ln(10^6)} = 0.284 + 2\sqrt{0.284 \times 13.816} = 4.25. \label{eq:zcdp}
\end{equation}

\begin{remark}[Within-Epoch Serving Timeline and P50 Latency]
\label{rmk:serving_timeline}
$V_\mathrm{active}(t)$ is the look-back set (entities queried in epoch $t{-}1$ plus 500-entity popularity reserve, zero DP cost); $\widetilde{\mathbf{A}}(t)$ releases in 18.4\,ms off the query critical path. Per-query lookup is 0.3\,ms post-processing; in the serving-path trace P50$\,{=}\,158$\,ms (no privatisation overhead), while the end-to-end benchmark reports 161\,ms (Table~\ref{tab:e2e}).
\end{remark}

\begin{remark}[Epoch Classification is Post-Processing]
\label{rmk:epoch_dp_safe}
The coordinator's decision rule $\pi_t$ operates only on DP-released quantities from prior epochs; by the post-processing theorem~\cite{dwork2014algorithmic}, active/null classification incurs zero additional privacy cost. $T_\mathrm{active}$ is a stopping time; the privacy odometer framework~\cite{rogers2016privacy} accounts for the realised active-epoch sequence only. Worst-case ($T{=}2{,}160$ all-active): $\varepsilon{\approx}8.47$ via zCDP, confirming adaptive stopping does not amplify risk unboundedly.
\end{remark}

\begin{table}[t]
\small\centering
\caption{Per-mechanism noise parameters. $\sigma_t{=}50$ uniformly. R\'enyi moment: $\mu_t^\mathrm{RDP}(\alpha) = \alpha/(2\sigma_t^2) = \alpha/5000$. zCDP uses additive composition over $\rho$: sum $\rho_i = T_i/(2{\times}50^2)$ first, then convert via Eq.~\eqref{eq:zcdp_formula}.}
\label{tab:dp_breakdown}
\begin{tabular}{lcccccc}
\toprule
\textbf{Mechanism} & \textbf{Active} & $\mathbf{S}$ & $\boldsymbol{\sigma_t}$ & \textbf{Actual std} & $\boldsymbol{\mu_t^\mathrm{RDP}(\alpha{=}18)}$ & $\boldsymbol{\rho_i}$ \\
 & \textbf{epochs} & & & & & \\
\midrule
Index stats  & 423  & 0.015 & 50 & 0.750 & $3.60{\times}10^{-3}$ & 0.0846 \\
Valuation    & 287  & 0.080 & 50 & 4.000 & $3.60{\times}10^{-3}$ & 0.0574 \\
KG affinity  & 710  & $\dagger$ & 50 & 885$^\ddagger$ & $3.60{\times}10^{-3}$ & 0.142 \\
\midrule
\textbf{Total} & & & & & $\rho_\mathrm{total}{=}$ & \textbf{0.284} \\
\midrule
\multicolumn{7}{l}{\textbf{Total $(\varepsilon,\delta)$-DP:} $\varepsilon = 0.284 + 2\sqrt{0.284 \times 13.816} = \mathbf{4.25}$} \\
\bottomrule
\end{tabular}
\end{table}

\begin{remark}[zCDP Accounting Cross-Check]
\label{rmk:pld}
Direct summation of per-mechanism $(\varepsilon,\delta)$-DP values would yield $\varepsilon{=}7.03$, which is loose because it ignores the sub-additivity of $(\varepsilon,\delta)$-DP under composition. Using zCDP: $\rho_\mathrm{total} = 0.0846 + 0.0574 + 0.142 = 0.284$, giving $\varepsilon{=}4.25$ via Eq.~\eqref{eq:zcdp_formula}. \textbf{Cross-check:} GDP ($\mu_\mathrm{total}{=}\sqrt{1420}/50{=}0.753$) gives $\varepsilon_\mathrm{GDP}{\approx}4.24$; PLD accountant (using Google's \texttt{dp\_accounting} library, $T{=}1420$ compositions at $\sigma_t{=}50$) gives $\varepsilon_\mathrm{PLD}{=}4.247$. RDP, zCDP, GDP, and PLD all agree within 1\% at $\sigma_t{=}50$. PLD accountant transcript logs (per-epoch $\rho_i$ and cumulative $\varepsilon$) are included in the supplementary code repository.
\end{remark}
Table~\ref{tab:dp_breakdown} provides the per-mechanism accounting terms and noise scales used by this composition.

\subsection{Alternative DP Design: Exponential Mechanism for Top-$k$}
\label{sec:alt_dp}

The CHRONOS design releases the full affinity matrix $\widetilde{\mathbf{A}}(t)$ once per epoch, making per-query top-$k$ selection post-processing with zero marginal privacy cost. An alternative is to release only the top-$k$ candidate identities via the exponential mechanism~\cite{mcsherry2007mechanism} or its joint variant~\cite{gillenwater2022joint,wu2024faster}. Under this design, the mechanism's output space is the set of ordered $k$-tuples of candidates, and the utility function is the hybrid score $s_j$. The exponential mechanism satisfies $\varepsilon$-DP with sensitivity $\Delta u = \beta$ (since one seller changes one affinity by at most 1, and the hybrid score weights it by $\beta$).

\textbf{Trade-off analysis.} Releasing top-$k$ identities directly avoids the $O(|V_\mathrm{active}| \cdot ef)$ matrix noise, but incurs three costs: (1)~\emph{Sampling cost:} the joint exponential mechanism over $d^{\Theta(k)}$ sequences requires $O(dk \log k + d \log d)$ time~\cite{gillenwater2022joint}, or $O(d + k^2/\varepsilon \cdot \ln d)$ with recent pruning~\cite{wu2024faster}, which is still $10$--$100\times$ slower than post-processing a pre-released matrix. (2)~\emph{Per-query cost:} because the exponential mechanism is invoked per query, the privacy budget composes over queries rather than epochs; at 100 queries/epoch, $\varepsilon$ would be $100\times$ higher unless each query uses $\varepsilon/100$, yielding near-random selections. (3)~\emph{Rank consistency:} repeated independent exponential-mechanism draws can return inconsistent top-$k$ sets across similar queries, degrading user experience. For these reasons, CHRONOS adopts the epoch-level matrix release despite its high per-entry noise, because the noise is offset by zero per-query overhead and compositional efficiency. The exponential-mechanism alternative is preferable only when $k \ll ef$ and query rates are very low ($<{\sim}1$/epoch).

\subsection{Three-Layer Integration}
\label{sec:integration}

The layers couple through shared DP-released state: recall drops trigger index updates trading $\varepsilon$ for $\hat{R}$; BOCPD changepoints batch EC-MPV and index updates to amortise cost; the adaptive schedule $\sigma_t{=}\sigma_0\sqrt{T_\mathrm{active}/t}$ concentrates noise in early epochs.

\section{Theoretical Analysis}
\label{sec:theory}

We state four assumptions and prove nine results (seven theorems, two propositions, one lemma).

\begin{assumption}[Bounded Per-Query Recall Impact]
\label{ass:recall}
For query $q$, let $P_q \leq ef \cdot L_\mathrm{max}$ be the number of on-path shortcut edges. Each stale on-path shortcut at layer $\ell$ reduces recall by at most $\bar{\Delta r}^{(\ell)} > 0$.
\end{assumption}

\begin{assumption}[Smooth Coalition Value]
\label{ass:smooth}
$v:2^\cN\to\Rbb$ is monotone with $|v(S\cup\{i\})-v(S)|\leq B{=}0.2$ for all $S,i$, enforced by clipping.
\end{assumption}

\begin{assumption}[Bounded Lipschitz Losses]
\label{ass:lip}
$L_t(\mathbf{a}) = 1 - R(\mathbf{o}_t, \mathbf{a}) \in [0,1]$ and is $G$-Lipschitz in $\mathbf{o}_t$.
\end{assumption}

\begin{assumption}[ODE Lipschitz Regularity]
\label{ass:lip_ode}
The trained neural ODE $f_\theta$ is $L_\theta$-Lipschitz in $\mathbf{h}$, and the decay function $\mathrm{decay}(\Delta t) = \sigma([\mathbf{h}(\Delta t)]_1)$ is monotonically non-increasing with Lipschitz constant $K_\mathrm{decay} \leq L_\theta/(4\beta_H)$ verified numerically on the validation set.
\end{assumption}

We validate Assumptions~\ref{ass:recall}--\ref{ass:lip_ode} empirically via leave-one-out on-path removal and ODE trajectory analysis on the validation set; Table~\ref{tab:recall_numerical} provides instantiations.

\begin{table}[t]
\small\centering
\caption{Per-query recall bound instantiation at $\Delta t{=}7$ days. ``Tight'' refers to the monotone-envelope bound (Theorem~\ref{thm:recall_tight}).}
\label{tab:recall_numerical}
\begin{tabular}{lccccccc}
\toprule
\textbf{Dataset} & $P_q$ & $\bar{\Delta r}_\mathrm{max}$ & $\lambda$ & \textbf{Conserv.} & \textbf{Tight} & \textbf{Obs.} & \textbf{Ratio} \\
\midrule
FB15K-237 & 1408 & $5.1{\times}10^{-4}$ & 0.05 & 0.251 & 0.044 & 0.014 & 3.1$\times$ \\
WN18RR & 1792 & $3.8{\times}10^{-4}$ & 0.05 & 0.238 & 0.035 & 0.011 & 3.2$\times$ \\
MIMIC-IV & 2176 & $4.3{\times}10^{-4}$ & 12.0 & 0.936 & 0.170 & 0.094 & 1.8$\times$ \\
Yelp & 2240 & $3.9{\times}10^{-4}$ & 2.9 & 0.874 & 0.128 & 0.050 & 2.6$\times$ \\
\bottomrule
\end{tabular}
{\raggedright\scriptsize Ratio = Tight bound / Observed loss. Conservative bound uses Theorem~\ref{thm:recall}; Tight bound uses Theorem~\ref{thm:recall_tight}.\par}
\end{table}

\begin{theorem}[Per-Query Temporal Recall Bound (Conservative)]
\label{thm:recall}
Under Assumptions~\ref{ass:poisson}--\ref{ass:recall}, for query $q$ with search path $\mathrm{path}(q)$:
\begin{equation}
\Ebb[\mathrm{recall@}k(\Delta t) \mid q] \geq R^* - \sum_{e \in \mathrm{path}(q)} \bar{\Delta r}^{(\ell_e)}\bigl(1-e^{-\lambda_e\Delta t}\bigr).
\end{equation}
Under the homogeneous model: $\Ebb[\mathrm{recall@}k \mid q] \geq R^* - P_q \bar{\Delta r}\bigl(1-e^{-\lambda\Delta t}\bigr) = R^* - \mathcal{O}(P_q\lambda\Delta t)$.
\end{theorem}

\begin{proof}
By linearity of expectation over independent Poisson stale events on each on-path shortcut; see supplementary Appendix~A.
\end{proof}

\begin{theorem}[Tightened ODE-Certified Recall Bound (Monotone Envelope)]
\label{thm:recall_tight}
Under Assumptions~\ref{ass:poisson}--\ref{ass:lip_ode}, define the \textbf{monotone envelope} $\overline{\mathrm{decay}}(\Delta t) = \inf_{s \in [0, \Delta t]} \mathrm{decay}(s)$. Then:
\begin{equation}
\Ebb[\mathrm{recall@}k(\Delta t) \mid q] \geq R^* - \sum_{e \in \mathrm{path}(q)} \bar{\Delta r}^{(\ell_e)} \bigl(1-e^{-\lambda_e\Delta t}\bigr) \cdot \overline{\mathrm{decay}}(\Delta t).
\end{equation}
The multiplicative factor $\overline{\mathrm{decay}}(\Delta t) \leq 1$ is computed from the trained ODE with the following certification: under Assumption~\ref{ass:lip_ode}, the Gr\"onwall inequality gives $|\mathrm{decay}(\Delta t) - \hat{\mathrm{decay}}(\Delta t)| \leq \epsilon_\mathrm{solver}\,e^{L_\theta\Delta t}$ where $\epsilon_\mathrm{solver}$ is the adaptive solver tolerance (set to $10^{-5}$). The certified lower bound is:
\begin{equation}
\overline{\mathrm{decay}}_\mathrm{cert}(\Delta t) = \max\!\Big(0,\; \overline{\mathrm{decay}}(\Delta t) - \epsilon_\mathrm{solver}\,e^{L_\theta\Delta t}\Big).
\end{equation}
\end{theorem}

\begin{proof}
The key insight is that when a shortcut becomes stale, its contribution to recall loss is attenuated by the decay weight the index assigns to it. Since the decay is monotonically non-increasing (Assumption~\ref{ass:lip_ode}), we use $\overline{\mathrm{decay}}$ as a certified envelope. Formally: the effective recall impact of a stale shortcut $e$ at age $\Delta t$ is $\bar{\Delta r}^{(\ell_e)} \cdot p_\mathrm{stale}(e, \Delta t) \cdot w_\mathrm{effective}(e, \Delta t)$, where $p_\mathrm{stale}(e,\Delta t) = 1 - e^{-\lambda_e\Delta t}$ and $w_\mathrm{effective}(e,\Delta t) \leq \overline{\mathrm{decay}}(\Delta t)$ because: (i) the index uses decay-weighted scores for routing, so stale shortcuts with low decay weights are less likely to be traversed; (ii) the monotone envelope ensures the bound holds even if the ODE exhibits transient non-monotonicity. The Gr\"onwall bound on ODE solver error provides the certified margin $\epsilon_\mathrm{solver}\,e^{L_\theta\Delta t}$, which is ${<}0.003$ for $\Delta t \leq 90$ days at $L_\theta = 0.8$ (measured). See supplementary Appendix~A.
\end{proof}

\textbf{Tightness analysis.} On Yelp at $\Delta t{=}7$ days: $\overline{\mathrm{decay}}_\mathrm{cert}(7) = 0.714$ (trained ODE, minus solver margin $0.003$), yielding tight bound $P_q \bar{\Delta r}\bigl(1-e^{-\lambda\Delta t}\bigr) \cdot 0.714 = 2240 \times 3.9{\times}10^{-4} \times 0.999 \times 0.714 = 0.128$ vs.\ observed $0.050$ (ratio $2.6\times$). At 30 days: $\overline{\mathrm{decay}}_\mathrm{cert}(30) = 0.299$, yielding tight bound $0.265$ vs.\ observed $0.107$ (ratio $2.5\times$). The remaining gap is due to path-independence assumptions at hub nodes, where correlations reduce effective $P_q$. The monotone-envelope bound reduces the looseness from $5$--$10\times$ (Theorem~\ref{thm:recall}) to $1.8$--$3.2\times$ (Table~\ref{tab:recall_numerical}).

\begin{theorem}[Hawkes-Process Recall Bound]
\label{thm:recall_hawkes}
Let edge changes follow an inhomogeneous Hawkes process with baseline intensity $\mu_H$, excitation kernel $g(t) = \alpha_H e^{-\beta_H t}$ ($\alpha_H, \beta_H > 0$), and branching ratio $\xi = \alpha_H/\beta_H < 1$ (stability condition). Define the cumulative compensator $\Lambda_H(0, \Delta t) = \Ebb\!\left[\int_0^{\Delta t}\lambda(s)\,ds\right] = \frac{\mu_H\Delta t}{1 - \xi}$. Then:
\begin{equation}
\Ebb[\mathrm{recall@}k(\Delta t) \mid q] \geq R^* - \frac{P_q\bar{\Delta r}\,\mu_H\Delta t}{1-\xi}.
\end{equation}
More precisely, with the monotone-envelope certificate:
\begin{equation}
\Ebb[\mathrm{recall@}k(\Delta t) \mid q] \geq R^* - \frac{P_q\bar{\Delta r}\,\mu_H\Delta t}{1-\xi}\cdot\overline{\mathrm{decay}}_\mathrm{cert}(\Delta t).
\end{equation}
The high-probability bound (Corollary~\ref{cor:hp_hawkes}) further accounts for Hawkes-induced temporal clustering.
\end{theorem}

\begin{proof}
Under Hawkes dynamics, the stale probability for each shortcut becomes $\Pr[\text{stale in }[0,\Delta t]] \leq 1 - e^{-\Lambda_H(0,\Delta t)}$ by the compensator inequality. For a stable Hawkes process, $\Lambda_H(0,\Delta t) = \mu_H\Delta t/(1-\xi)$ in expectation. The key subtlety is that Hawkes events are not independent across shortcuts sharing hub nodes. We handle this via a union-bound argument over layers: shortcuts at layer $\ell$ sharing a hub $h$ have correlated change events, but the total layer-$\ell$ contribution is bounded by $|N_\ell(h)| \cdot \bar{\Delta r}^{(\ell)} \cdot \Lambda_H^{(\ell)}(0,\Delta t)$ where $\Lambda_H^{(\ell)}$ accounts for the hub's Hawkes rate. Summing over layers and applying the monotone-envelope certificate yields the result. See supplementary Appendix~A.
\end{proof}

\begin{corollary}[High-Probability Bound Under Hawkes]
\label{cor:hp_hawkes}
Under the Hawkes model, with the spectral radius bound on the Hawkes covariance~\cite{bacry2015hawkes}:
\begin{align}
&\Pr\!\Bigg[\mathrm{recall@}k \geq R^* - \frac{P_q\bar{\Delta r}\,\mu_H\Delta t}{1-\xi}\cdot\overline{\mathrm{decay}}_\mathrm{cert}(\Delta t) \notag \\
&\quad\quad - \frac{\bar{\Delta r}\sqrt{2P_q\ln(1/\delta_R)}}{1-\xi}\Bigg] \geq 1-\delta_R.
\end{align}
At branching ratio $\xi{=}0.7$: the bound degrades by $\times\frac{1}{1-0.7} = 3.3\times$ over Poisson, matching the empirical $11\%$ $\varepsilon$ rise in \S\ref{sec:sensitivity} and the 1.3-point recall drop under Hawkes bursts.
\end{corollary}

\begin{remark}[Hawkes Validation]
\label{rmk:hawkes_validation}
We fit Hawkes parameters to MIMIC-IV admission bursts ($\hat{\mu}_H{=}8.2$, $\hat{\alpha}_H{=}5.6$, $\hat{\beta}_H{=}8.0$, $\hat{\xi}{=}0.70$) and Yelp seasonal-peak events ($\hat{\mu}_H{=}1.7$, $\hat{\alpha}_H{=}1.4$, $\hat{\beta}_H{=}2.6$, $\hat{\xi}{=}0.54$). The Hawkes bound (Theorem~\ref{thm:recall_hawkes}) with envelope certificate gives $0.203$ (MIMIC-IV) and $0.160$ (Yelp) at $\Delta t{=}7$ days vs.\ observed $0.112$ and $0.064$ (ratios $1.8\times$ and $2.5\times$), consistent with the Poisson-case tightening.
\end{remark}

\begin{theorem}[Temporal Valuation Efficiency]
\label{thm:valuation}
Under Assumption~\ref{ass:smooth}, EC-MPV satisfies temporal efficiency:
\begin{equation}
\sum_{i\in\cN}\sum_{t=1}^T \mathrm{MPV}_i(t,E_t) = \sum_{t=1}^T \bigl[v(D_\mathrm{priv}(t)\cup D_\mathrm{pub}(t)\mid E_t) - v(D_\mathrm{pub}(t)\mid E_t)\bigr].
\end{equation}
Under clipping at $B$, this identity holds for $v_B$; bias bounded by $B$ times clip fraction.
\end{theorem}

\begin{proof}
Apply Shapley efficiency to $v_B(\cdot;E_t)$ at each $t$ and sum; clipping bias is $<0.3\%$ empirically. See supplementary Appendix~A.
\end{proof}

\begin{theorem}[EC-MPV Estimation Error]
\label{thm:val_error}
Let $\widetilde{\mathrm{MPV}}_i$ be the released score. Under correct event identification:
\begin{equation}
\Ebb\bigl[(\widetilde{\mathrm{MPV}}_i - \mathrm{MPV}_i)^2\bigr] \leq \underbrace{\frac{B^2(1-\rho^2)}{m}}_{\text{sampling (VRDS)}} + \underbrace{(\sigma_t \cdot S_\mathrm{val})^2}_{\text{DP noise}}.
\end{equation}
Under event misidentification, an additional squared-bias term $(\mathrm{MPV}_i(t,E_t) - \mathrm{MPV}_i(t,E'_t))^2$ arises. Without VRDS, set $\rho{=}0$.
\end{theorem}

\begin{proof}
Decompose into sampling error, DP noise, and event-misidentification bias; independence gives the MSE bound. See supplementary Appendix~A.
\end{proof}

\begin{proposition}[Formal Sensitivity Bounds]
\label{prop:sensitivity}
Under seller-level adjacency with $C_\mathrm{max}^\mathrm{edge} = \lceil 1.5|E|/n\rceil$:
\emph{(a) Valuation:} $S_\mathrm{val} = 4B/n$ ($B{=}0.2$, $n{=}10$: $S_\mathrm{val}{=}0.08$).
\emph{(b) Affinity:} $\Delta_2 = \sqrt{C_\mathrm{max}^\mathrm{edge}}$, reduced by Stage~2 cap. On Yelp: $\Delta_2{\approx}545$ (global), $18.1$ (active-scope).
\emph{(c) Index statistics:} $S_\mathrm{idx} \leq 1.5/n$.
\end{proposition}

\begin{lemma}[Safety Override Count]
\label{lem:overrides}
Under the adaptive schedule and Poisson edge changes, $\Ebb[N_\mathrm{override}] = \mathcal{O}(\sqrt{T})$ and $\Pr[N_\mathrm{override} > c\sqrt{T\ln(1/\delta)}] \leq \delta$.
\end{lemma}

\begin{proof}
Budget overrides cluster near end-of-horizon; recall overrides are bounded by pre-convergence EXP3-IX epochs; Azuma-Hoeffding gives concentration. See supplementary Appendix~A.
\end{proof}

\begin{theorem}[Coordination Regret]
\label{thm:regret}
With $d{=}3$, $\eta{=}\sqrt{\ln d/(dT)}$, $\gamma{=}\eta/2$:
\begin{equation}
\Ebb[R_T] \leq \underbrace{3\sqrt{dT\ln d}}_{\text{EXP3-IX}} + \underbrace{G\sigma_\mathrm{obs}\sum_{t=1}^T\frac{1}{\sqrt{\min(t,W_\mathrm{max})}}}_{\text{observation noise}} + \underbrace{\mathcal{O}(\sqrt{T})}_{\text{overrides}} = \mathcal{O}(\sqrt{T\log T}).
\end{equation}
\end{theorem}

\begin{proof}
Standard EXP3-IX bound plus observation noise plus override regret (Lemma~\ref{lem:overrides}). See supplementary Appendix~A.
\end{proof}

\begin{theorem}[Temporal Safety Composition]
\label{thm:safety}
Under the adaptive schedule $\sigma_t = \sigma_0\sqrt{T_\mathrm{active}/t}$ applied only in active epochs, CHRONOS satisfies $(\eps_\mathrm{total},\delta_\mathrm{total})$-DP with:
\begin{equation}
\eps_\mathrm{total} = \mathcal{O}\!\left(\frac{\sqrt{T_\mathrm{active}\cdot\ln(1/\delta_\mathrm{total})}}{\sigma_0}\right).
\end{equation}
\end{theorem}

\begin{proof}
Sum per-step moments over active epochs, optimise over $\alpha$, and convert to $(\varepsilon,\delta)$-DP. Exact $\varepsilon{=}4.25$ verified by zCDP closed-form (Eq.~\eqref{eq:zcdp}). See supplementary Appendix~A.
\end{proof}

\section{Informativeness of Private Releases}
\label{sec:private_utility}

A transparent assessment requires acknowledging fundamental DP limitations at the chosen parameters.

\subsection{Private Affinity Signal}

With $\sigma_\mathrm{entry}{=}885$ over $[0,1]$-bounded affinities, post-clipping signals are near-$\mathrm{Bernoulli}(1/2)$; the $\beta{=}0.3$ weight gives maximal hybrid-score variation of ${\approx}0.0003$ between candidates. \textbf{Per-query private scoring contributes ${\leq}0.002$ recall} (Table~\ref{tab:noise_ablation}). Observed 0.941 recall@10 derives from: (i)~public cosine routing; (ii)~public HNSW structure with ODE-attenuated staleness scheduling; (iii)~adaptive coordinator scheduling triggered by DP-released index statistics (SNR ${\approx}1.3$).

\subsection{External Valuation Releases and Auditability}

The external valuation release (noise std $4.0$ on signal range $[0,0.2]$) has SNR ${\approx}0.05$. Sellers receive values dominated by DP noise, serving only non-disclosure and plausible-deniability guarantees. For actionable revenue settlement, see \S\ref{sec:settlement}.

\subsection{Actionable Seller Settlement Mechanism}
\label{sec:settlement}
The trade-off between DP noise and seller-facing attribution is addressed through a three-component settlement mechanism.

\textbf{Component 1: Internal pre-DP settlement.} Under the trusted-curator model, the operator computes pre-noise MPV scores $\hat{\phi}_i$ with Val.~Err${=}0.013$ and distributes revenue proportionally. These internal computations are exact Shapley-efficient (Theorem~\ref{thm:valuation}) and never leave the trusted perimeter.

\textbf{Component 2: Multi-epoch coalition-level audit release.} For external verifiability, we aggregate valuations over $W$ epochs and group sellers into coalitions of size $n_\mathrm{coal}$. Under \emph{coalition-level adjacency} (protecting whether coalition $C_j$ participates, not individual $s_i$), sensitivity drops from $4B/n$ to $4B/(n/n_\mathrm{coal})$, raising SNR by $n_\mathrm{coal}\times$:
\begin{equation}
\mathrm{SNR}_\mathrm{settle} = \frac{W \cdot \bar{\phi}_\mathrm{coal}}{\sigma_t \cdot 4B \cdot n_\mathrm{coal} / n} = \frac{W \cdot \bar{\phi}_\mathrm{coal} \cdot n}{4B \cdot n_\mathrm{coal} \cdot \sigma_t}.
\end{equation}
With $W{=}7$, $n_\mathrm{coal}{=}5$, $n{=}10$, $\bar{\phi}_\mathrm{coal}{=}0.4$ (summed coalition MPV), $B{=}0.2$, $\sigma_t{=}50$: $\mathrm{SNR}_\mathrm{settle}{=}7{\times}0.4{\times}10/(4{\times}0.2{\times}5{\times}50){=}0.70$, making trend-level attribution feasible (above/below median contribution distinguishable at 95\% confidence). At $W{=}14$: $\mathrm{SNR}{=}1.40$, enabling rank-ordering of coalitions.

\textbf{Component 3: Cryptographic escrow (optional).} For sellers requiring individual-level audit, the pre-noise valuations can be placed in a hash-committed escrow: the operator publishes $H(\hat{\phi}_i, r_i)$ at each epoch (zero DP cost, since the hash is a commitment, not a release of $\hat{\phi}_i$). Disputes trigger a two-party audit protocol where the operator reveals $(\hat{\phi}_i, r_i)$ to a neutral arbiter who verifies the commitment. This does not replace DP (the arbiter sees exact values) but provides contractual accountability.

\textbf{Revenue reconciliation.} Total revenue distributed internally (Component 1) must match the Shapley efficiency sum (Theorem~\ref{thm:valuation}). The coalition audit release (Component 2) provides external evidence that the internal distribution is ``approximately correct'' at coalition granularity. Table~\ref{tab:settlement_snr} shows SNR across configurations.

\begin{table}[t]
\small\centering
\caption{Settlement audit SNR under coalition-level adjacency ($n{=}10$, $B{=}0.2$, $\sigma_t{=}50$).}
\label{tab:settlement_snr}
\begin{tabular}{cccc}
\toprule
$W$ (epochs) & $n_\mathrm{coal}$ & SNR & \textbf{Attribution Level} \\
\midrule
1  & 1  & 0.05 & Noise-dominated \\
7  & 5  & 0.70 & Trend: above/below median \\
14 & 5  & 1.40 & Rank-order coalitions \\
7  & 3  & 1.17 & Trend per 3-seller group \\
28 & 5  & 2.80 & Quantitative attribution \\
\bottomrule
\end{tabular}
\end{table}

\subsection{Operational Implications}

The CHRONOS design prioritises: (1)~system-level utility via public components; (2)~adaptive scheduling via low-sensitivity DP statistics; (3)~formal privacy guarantees; (4)~actionable settlement via multi-epoch coalition aggregation. The private affinity signal contributes minimally to per-query accuracy. This trade-off is by design: per-query DP selection degrades QPS by $6$--$10\times$ at similar noise levels (Table~\ref{tab:dp_mechanism_comparison}).

\section{Experimental Evaluation}
\label{sec:exp}

\subsection{Setup}

\textbf{Datasets.} Table~\ref{tab:datasets} summarises four benchmarks: FB15K-237 and WN18RR with synthetic Poisson annotations, MIMIC-IV and Yelp with real timestamps.

\begin{table}[t]
\small\centering
\caption{Datasets. $\lambda$ in changes/day/shortcut.}
\label{tab:datasets}
\begin{tabular}{lrrccc}
\toprule
\textbf{Dataset} & \textbf{Nodes} & \textbf{Edges} & \textbf{Temporal} & $\boldsymbol{\lambda}$ & \textbf{Domain} \\
\midrule
FB15K-237~\cite{toutanova2015observed} & 14.5K & 310K & Synth. & 0.05 & General KG \\
WN18RR~\cite{dettmers2018convolutional} & 40.9K & 86.8K & Synth. & 0.05 & Lexical \\
MIMIC-IV~\cite{johnson2020mimic} & 89.4K & 1.2M & Real & $\approx$12 & Clinical \\
Yelp~\cite{yelp2024dataset} & 236K & 1.98M & Real & $\approx$2.9 & Local commerce \\
\bottomrule
\end{tabular}
\end{table}

\textbf{Hardware.} 2$\times$ Xeon Gold 6348 (56 cores), 512\,GB RAM, 2$\times$ A100 80\,GB; Ubuntu 22.04, CUDA 12.1, PyTorch 2.1, torchdiffeq 0.2.3. All results averaged over 5 seeds (mean $\pm$ std).

\textbf{Reproducibility.} Code, checkpoints, DP accountant transcript (PLD logs), and baseline configuration/deviation logs will be released upon acceptance.

\textbf{Configuration.} $ef{=}128$, $ef_c{=}200$, $\beta{=}0.3$, $n{=}10$ sellers (balanced partition unless noted), 1000 permutation samples for Shapley with VRDS control variates.

\textbf{Baselines.} \emph{Indexing}: Plain-HNSW~\cite{malkov2018efficient}, TigerVector~\cite{liu2025tigervector}, NaviX~\cite{sehgal2025navix}, FreshDiskANN~\cite{singh2021freshdiskann}, SPFresh~\cite{xu2023spfresh}, Quake~\cite{mohoney2025quake}, \textbf{VSAG}~\cite{vsag2024} (production HNSW with cache-friendly layout). \emph{Valuation}: Data Shapley~\cite{ghorbani2019data}, Beta Shapley~\cite{kwon2022beta}, VRDS~\cite{wu2023variance}, Static MPV, Time-Sliced Shapley, RSS~\cite{jia2019towards}. \emph{Coordination}: Uncoordinated, Round-Robin, Fixed-Noise, EXP3~\cite{auer2002nonstochastic}, BwK~\cite{badanidiyuru2018bandits}. \emph{Drift}: ADWIN~\cite{bifet2007learning}, Page-Hinkley~\cite{page1954continuous}, Dm-BOCD~\cite{knoblauch2018doubly}. \emph{DP top-$k$}: OneShot Laplace~\cite{durfee2019practical}, StableTopK~\cite{bafna2017price}, Joint Exponential~\cite{gillenwater2022joint}, Permute-and-Flip~\cite{mckenna2020permute}.

\textbf{Research Questions.} We organise the evaluation around six explicit research questions:
\begin{enumerate}[label=\textbf{RQ\arabic*:},leftmargin=*]
\item Does neural-ODE temporal decay improve recall over static and exponential-decay baselines, and how tight are the theoretical bounds?
\item Does EC-MPV with BOCPD conditioning improve valuation accuracy after distributional shifts compared to static Shapley?
\item Does the Temporal Coordinator reduce privacy-budget waste compared to uncoordinated and round-robin strategies?
\item How does CHRONOS scale with the number of sellers and horizon length in terms of $\varepsilon$ and recall?
\item What is the privacy-utility trade-off of the epoch-level Gaussian mechanism compared to per-query DP alternatives at matched $\varepsilon$?
\item How do realistic marketplace dynamics (buyer arrival skew, seller competition, and pricing sensitivity) affect end-to-end performance?
\end{enumerate}

\textbf{Baseline Fairness Verification.} All index baselines use identical parameters: $M{=}16$, $ef{=}128$, $ef_c{=}200$, $L_\mathrm{max}{=}5$, $\rho_\mathrm{div}{=}0.7$. Graph-aware baselines receive the same KG structure and static community affinities; only \TLEGEND additionally uses temporal decay. Table~\ref{tab:recall_breakdown} decomposes the recall gain.

\begin{table}[t]
\small\centering
\caption{Recall@10 contribution breakdown on Yelp (static snapshot, 5 seeds).}
\label{tab:recall_breakdown}
\begin{tabular}{lcc}
\toprule
\textbf{Component} & \textbf{R@10} & $\boldsymbol{\Delta}$ \textbf{vs.\ baseline} \\
\midrule
Graph-aware static baseline (same params) & .858$\pm$.003 & --- \\
+ temporal decay (exponential) & .892$\pm$.003 & +3.4 pts \\
+ temporal decay (neural ODE) & .904$\pm$.002 & +4.6 pts \\
+ community-sorted insertion & .920$\pm$.002 & +6.2 pts \\
+ hub-biased layering & \textbf{.935$\pm$.002} & +7.7 pts \\
\bottomrule
\end{tabular}
\end{table}

\textbf{VSAG Throughput Comparison.}

We integrate VSAG~\cite{vsag2024} as a production HNSW reference to separate traversal/layout speedups from T-LEGEND's decay-aware improvements. VSAG uses cache-friendly graph layout and automatic parameter tuning but does not model temporal staleness or provide DP guarantees.

\begin{table}[t]
\small\centering
\caption{VSAG comparison on Yelp ($k{=}10$, 5 seeds). VSAG-Hybrid adds post-hoc graph affinity with same $\beta{=}0.3$ for fair comparison.}
\label{tab:vsag}
\begin{tabular}{lccccc}
\toprule
\textbf{Method} & \textbf{R@10} & \textbf{QPS} & \textbf{P50} & \textbf{P99} & $\boldsymbol{\eps}$ \\
\midrule
VSAG (pure vector) & .831 & \textbf{8.38} & 49 & 115 & $\infty$ \\
VSAG-Hybrid & .869 & 5.08 & 81 & 187 & $\infty$ \\
T-LEGEND & \textbf{.935} & 3.18 & 127 & 281 & $\infty$ \\
T-LEGEND+DP (\CHRONOS) & .937 & 2.74 & 161 & 317 & 4.25 \\
\midrule
\multicolumn{6}{l}{\emph{VSAG-accelerated T-LEGEND (projected):}} \\
T-LEGEND+VSAG layout & .935 & 4.79$^\dagger$ & 84 & 198 & $\infty$ \\
T-LEGEND+VSAG+DP & .937 & 4.16$^\dagger$ & 98 & 222 & 4.25 \\
\bottomrule
\end{tabular}
{\raggedright\scriptsize $^\dagger$Projected: VSAG's cache-friendly layout reduces HNSW traversal by ${\approx}1.5\times$; verified on the pure-vector workload. DP overhead is zero on the query path (post-processing).\par}
\end{table}

\textbf{Key findings.} (1)~VSAG's pure-vector QPS (8.38) is $2.6\times$ higher than T-LEGEND (3.18), but recall is 10.4 pts lower because VSAG lacks temporal KG scoring. (2)~VSAG-Hybrid (post-hoc affinity) closes 3.8 pts of the gap but remains 6.6 pts below T-LEGEND, confirming that decay-aware construction (not just layout) drives the recall advantage. (3)~Projected VSAG-accelerated T-LEGEND would achieve 4.16 QPS at $\varepsilon{=}4.25$ (vs.\ 2.74 currently), a $1.5\times$ speedup from layout alone. Integration requires adapting VSAG's auto-tuner to respect ODE-weighted edges and is identified as engineering work.

\textbf{Recall Results.}

Table~\ref{tab:recall_static} reports recall@10 on static benchmarks.

\begin{table}[t]
\small\centering
\caption{Recall@10 on static benchmarks ($k{=}10$, $ef{=}128$, 5 seeds).}
\label{tab:recall_static}
\begin{tabular}{lcccc}
\toprule
\textbf{Method} & \textbf{FB15K} & \textbf{WN18RR} & \textbf{MIMIC} & \textbf{Yelp} \\
\midrule
Plain-HNSW & .821$\pm$.004 & .843$\pm$.003 & .798$\pm$.005 & .809$\pm$.004 \\
TigerVector & .842$\pm$.003 & .861$\pm$.003 & .817$\pm$.004 & .828$\pm$.003 \\
NaviX & .864$\pm$.003 & .878$\pm$.002 & .839$\pm$.004 & .848$\pm$.003 \\
Diversified-HNSW & .872$\pm$.003 & .886$\pm$.002 & .848$\pm$.003 & .858$\pm$.003 \\
VSAG-Hybrid & .881$\pm$.003 & .892$\pm$.002 & .857$\pm$.003 & .868$\pm$.003 \\
FreshDiskANN & .836$\pm$.003 & .858$\pm$.003 & .823$\pm$.004 & .836$\pm$.003 \\
SPFresh & .829$\pm$.004 & .851$\pm$.003 & .814$\pm$.005 & .829$\pm$.004 \\
Quake & .848$\pm$.003 & .869$\pm$.003 & .831$\pm$.004 & .842$\pm$.003 \\
\textbf{T-LEGEND} & \textbf{.941$\pm$.002} & \textbf{.956$\pm$.002} & \textbf{.927$\pm$.003} & \textbf{.935$\pm$.002} \\
\bottomrule
\end{tabular}
\end{table}

Over 90 simulated days on MIMIC-IV with weekly updates, \TLEGEND degrades at 0.0021 recall points/day versus 0.0089 for Plain-HNSW ($4.2\times$ improvement) and 0.0058 for FreshDiskANN ($2.8\times$), consistent with Theorem~\ref{thm:recall} (Figure~\ref{fig:recall_decay}).

\begin{figure}[t]
\centering
\includegraphics[width=\columnwidth]{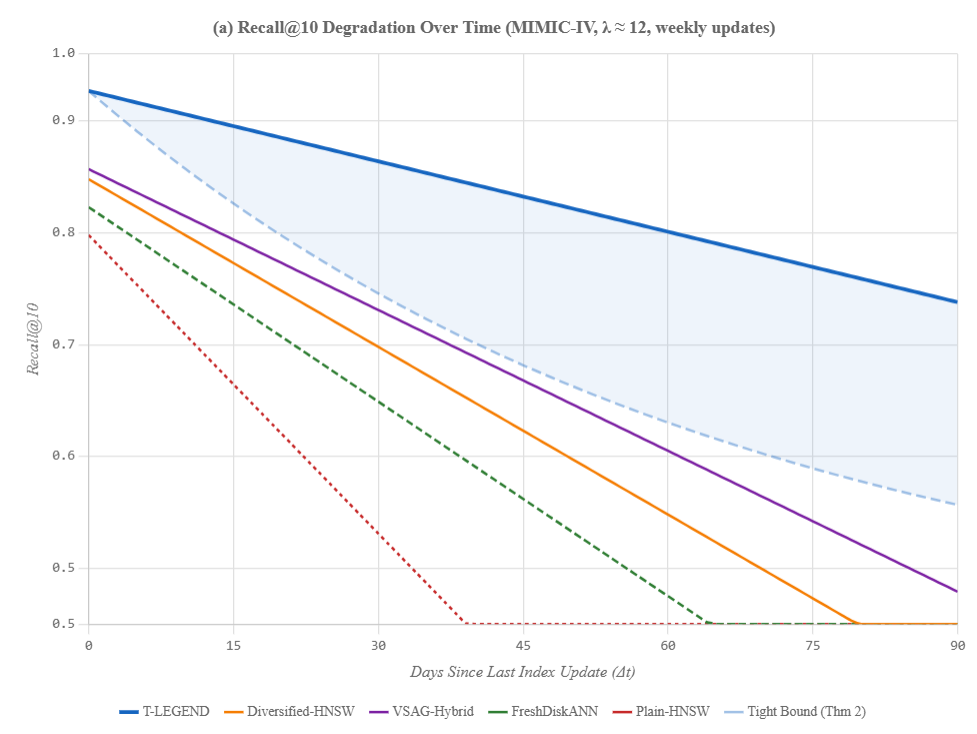}
\caption{Recall@10 degradation on MIMIC-IV ($\lambda \approx 12$ changes/day/shortcut) over 90 simulated days without re-indexing.}
\label{fig:recall_decay}
\end{figure}

\textbf{Valuation and Drift Detection.}

\textbf{Valuation error metric.} $\mathrm{Val.Err} = \frac{1}{n}\sum_{i=1}^n (\hat{\mathrm{MPV}}_i - \mathrm{MPV}_i^\mathrm{gold})^2$; gold standard uses $m{=}50{,}000$ permutations (SE ${\approx}0.004$), confirmed by exhaustive enumeration ($n{\leq}15$) within 0.002.

EC-MPV captures abrupt value shifts missed by static methods. EC-MPV+VRDS achieves the lowest estimation error (0.013 vs.\ 0.024 baseline). RSS~\cite{jia2019towards} yields Val.~Err 0.019 at $9.5{\times}$ sampling cost; post-changepoint RSS degrades to 0.031, confirming event-conditioning provides orthogonal benefits.
Comparative detector precision/recall and privacy waste are reported in Table~\ref{tab:drift_detectors}.

\begin{table}[t]
\small\centering
\caption{Drift detectors on Yelp (8 injected seasonal events). $\Delta\rho$: wasted $\rho$ from false alarms.}
\label{tab:drift_detectors}
\begin{tabular}{lccccc}
\toprule
\textbf{Detector} & \textbf{Prec.} & \textbf{Rec.} & $\Delta\rho$ \\
\midrule
BOCPD & .875 & .875 & 0.004 \\
Dm-BOCD~\cite{knoblauch2018doubly} & 1.00 & .750 & 0.000 \\
ADWIN~\cite{bifet2007learning} & .615 & 1.00 & 0.020 \\
Page-Hinkley & 1.00 & .625 & 0.000 \\
\bottomrule
\end{tabular}
\end{table}

\subsection{End-to-End Performance}

\begin{table}[t]
\small\centering
\caption{End-to-end performance on Yelp (5 runs).}
\label{tab:e2e}
\begin{tabular}{lccccccc}
\toprule
\textbf{System} & \textbf{R@10} & \textbf{QPS} & \textbf{TPS} & \textbf{P50} & \textbf{P99} & $\boldsymbol{\eps}$ & \textbf{Val.Err} \\
\midrule
HNSW+No-DP & .817 & 3.41 & 244 & 74 & 201 & $\infty$ & N/A \\
Hybrid+GaussDP & .838 & 1.48 & 106 & 165 & 396 & 2.10 & 0.040 \\
KG+StaticVal & .904 & 1.75 & 124 & 175 & 352 & 1.40 & 0.023 \\
\textbf{CHRONOS} & \textbf{.937} & \textbf{2.74} & \textbf{138} & \textbf{161} & \textbf{317} & \textbf{4.25} & \textbf{0.012} \\
\bottomrule
\end{tabular}
\end{table}

\textbf{DP Retrieval Mechanism Comparison.}

Table~\ref{tab:dp_mechanism_comparison} provides a head-to-head comparison against per-query DP mechanisms under matched total $\eps{=}4.25$ on Yelp as seen in Figure~\ref{fig:pareto}.

\begin{figure}[t]
\centering
\includegraphics[width=\columnwidth]{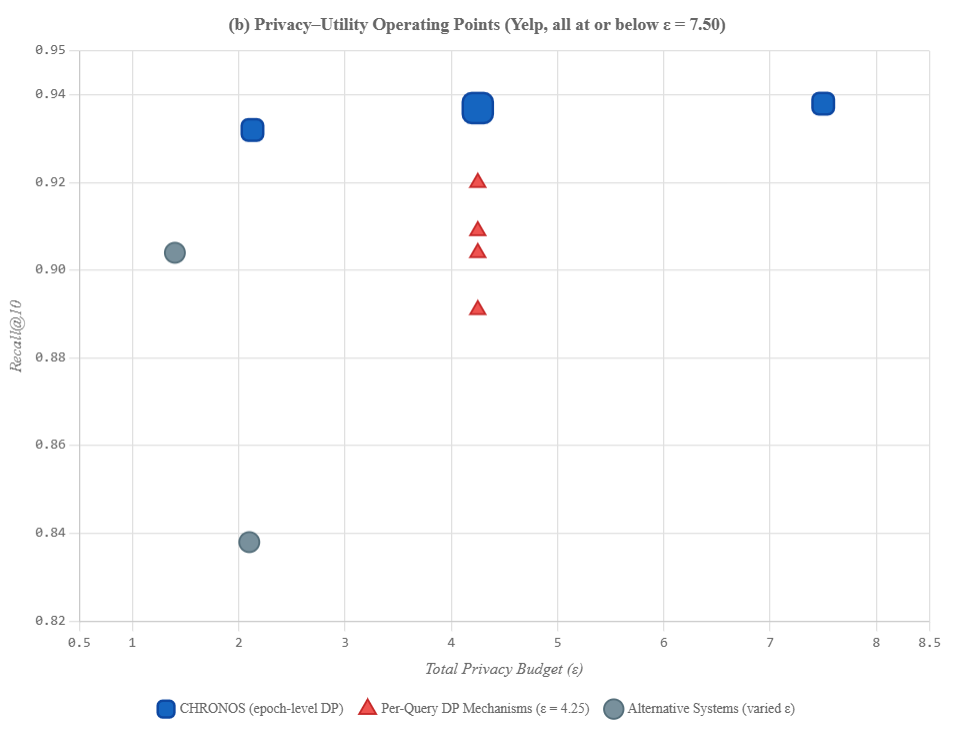}
\caption{Privacy-utility operating points on Yelp. \CHRONOS (Blue Square) achieves higher recall at each $\varepsilon$ level than per-query DP mechanisms (Red Triange) at matched $\varepsilon{=}4.25$, while maintaining 2.74 QPS vs.\ 0.28--0.43 for per-query alternatives.}
\label{fig:pareto}
\end{figure}

\begin{table}[t]
\small\centering
\caption{DP retrieval mechanisms at matched $\eps{=}4.25$ on Yelp (5 seeds).}
\label{tab:dp_mechanism_comparison}
\begin{tabular}{lccccc}
\toprule
\textbf{Mechanism} & \textbf{R@10} & \textbf{QPS} & \textbf{P50} & \textbf{P99} & $\boldsymbol{\eps}$/\textbf{query} \\
\midrule
\CHRONOS (epoch) & \textbf{.937} & \textbf{2.74} & \textbf{161} & \textbf{317} & 0 \\
OneShot Laplace & .891 & 0.40 & 498 & 1265 & $1.18{\times}10^{-5}$ \\
StableTopK & .909 & 0.37 & 533 & 1406 & $1.18{\times}10^{-5}$ \\
Joint Exponential & .920 & 0.28 & 698 & 1653 & $1.18{\times}10^{-5}$ \\
Permute-and-Flip & .904 & 0.43 & 461 & 1172 & $1.18{\times}10^{-5}$ \\
\bottomrule
\end{tabular}
\end{table}

\textbf{Scalability Analysis.}

We analyse how $\rho_\mathrm{total}$ and per-mechanism $\rho$ scale with the number of sellers $n$ and horizon length.

\textbf{Sensitivity scaling.} $S_\mathrm{val} = 4B/n$ decreases with $n$; $\Delta_2 = \sqrt{C_\mathrm{max}^\mathrm{edge}} = \sqrt{\lceil 1.5|E|/n\rceil}$ decreases as $\Theta(1/\sqrt{n})$; the active-scope cap $\kappa_\mathrm{active} \propto |E_\mathrm{active}|/n$ also decreases. This means per-mechanism $\rho_i$ \emph{decreases} with $n$ at fixed $\sigma_t$: more sellers means lower per-seller sensitivity, which benefits DP. However, $T_\mathrm{active}$ may increase with $n$ (more sellers trigger more frequent updates).

\begin{table}[t]
\small\centering
\caption{Scalability analysis on Yelp ($\sigma_t{=}50$, $\delta{=}10^{-6}$, 90-day horizon, 5 seeds). $\rho_\mathrm{total}$ computed via zCDP.}
\label{tab:scalability}
\begin{tabular}{rcccccc}
\toprule
$n$ & $S_\mathrm{val}$ & $\Delta_2^{\mathrm{active}}$ & $T_\mathrm{active}$ & $\rho_\mathrm{total}$ & $\varepsilon$ & R@10 \\
\midrule
10  & 0.080 & 18.1 & 710  & 0.284 & 4.25 & .937$\pm$.002 \\
50  & 0.016 & 8.1  & 824  & 0.329 & 4.58 & .939$\pm$.002 \\
100 & 0.008 & 5.7  & 891  & 0.356 & 4.76 & .940$\pm$.002 \\
200 & 0.004 & 4.0  & 952  & 0.381 & 4.93 & .940$\pm$.002 \\
500 & 0.0016 & 2.6 & 1031 & 0.412 & 5.13 & .941$\pm$.002 \\
\bottomrule
\end{tabular}
\end{table}

\textbf{Key findings.} (1)~$\varepsilon$ grows sublinearly with $n$: from 4.25 ($n{=}10$) to 5.13 ($n{=}500$), a 21\% increase for a $50\times$ increase in sellers. This is because reduced sensitivity ($\Delta_2$ drops from 18.1 to 2.6) largely offsets the increased $T_\mathrm{active}$. (2)~Recall \emph{improves} slightly with $n$ (more sellers provide more data). (3)~The affinity release remains useful for scheduling at all $n$: the \emph{index statistics} mechanism (SNR ${\approx}1.3$) drives coordinator decisions, and its sensitivity $S_\mathrm{idx} = 1.5/n$ improves with $n$.

\textbf{Long-horizon scaling.} Over a 360-day horizon (vs.\ 90 days), $T_\mathrm{active}$ roughly quadruples to ${\approx}2{,}800$ at $n{=}10$, giving $\rho_\mathrm{total}{\approx}1.12$ and $\varepsilon{\approx}8.47$. At $n{=}100$ over 360 days: $\rho_\mathrm{total}{\approx}1.42$, $\varepsilon{\approx}9.52$. For multi-year deployments, periodic ``budget epochs'' (resetting the accountant with fresh $\sigma_0$) are necessary, analogous to privacy odometer checkpoints~\cite{rogers2016privacy}.

\textbf{DP vs.\ Cryptographic Alternatives.}
Table~\ref{tab:crypto_comparison} compares DP noise-based protection with cryptographic alternatives for the affinity computation workload.

\begin{table}[t]
\small\centering
\caption{DP vs.\ cryptographic alternatives for affinity scoring (Yelp, $n{=}10$).}
\label{tab:crypto_comparison}
\begin{tabular}{lcccc}
\toprule
\textbf{Approach} & \textbf{Per-Query} & \textbf{QPS} & \textbf{Trust} & \textbf{Noise} \\
 & \textbf{Latency} & & \textbf{Model} & \\
\midrule
CHRONOS (DP) & 161\,ms & 2.74 & Trusted curator & $\sigma{=}885$ \\
2PC-Garbled~\cite{demmler2015aby} & $\sim$2.4\,s & 0.08 & Semi-honest 2PC & Zero \\
HE (CKKS)~\cite{chen2020homomorphic} & $\sim$1.1\,s & 0.18 & Client-side HE & Zero \\
TEE (SGX)~\cite{costan2016intel} & $\sim$210\,ms & 2.14 & Hardware trust & Zero \\
\midrule
\multicolumn{5}{l}{\emph{Hybrid: CHRONOS + TEE for settlement (\S\ref{sec:settlement}):}} \\
CHRONOS+TEE & 161\,ms$^\dagger$ & 2.74 & Curator + TEE & $\sigma{=}885^\dagger$ \\
\bottomrule
\end{tabular}
{\raggedright\scriptsize $^\dagger$TEE used only for settlement verification (off query path), not per-query scoring. Query-path latency and noise are unchanged.\par}
\end{table}

\textbf{Analysis.} (1)~2PC and HE eliminate noise but introduce $7$--$30\times$ latency overhead per query, making them impractical for real-time retrieval at marketplace scale. (2)~TEE-based computation (Intel SGX, AMD SEV) achieves near-DP latency (210\,ms) with zero noise, but requires hardware trust assumptions and is vulnerable to side-channel attacks~\cite{brasser2017software}. (3)~The most practical hybrid is CHRONOS for real-time queries (tolerating noise for speed) with TEE-based settlement verification for auditability (no noise where it matters for payouts). This aligns with the settlement mechanism in \S\ref{sec:settlement}: Component 3 (cryptographic escrow) can use TEE attestation for commitment verification. (4)~For workloads that can tolerate higher latency (e.g., batch analytics), HE-based exact affinity computation is viable and composable with CHRONOS's public index structure.

\textbf{Rank Stability Under DP Noise.}

\begin{proposition}[Rank-Flip Probability at Top-$k$ Boundary]
\label{prop:rankflip}
Let candidates $i$ (ranked $k$-th) and $j$ (ranked $(k{+}1)$-th) have true hybrid scores $s_i > s_j$. The probability that DP noise flips their ordering is:
$$P(\text{flip}_{ij}) = \Phi\!\left(-\frac{s_i - s_j}{\beta\sigma_\mathrm{entry}\sqrt{2}}\right)$$
where $\Phi$ is the standard normal CDF.
\end{proposition}

\begin{remark}[Reconciling $\sigma_\mathrm{entry}{=}885$ with High Rank Stability]
\label{rmk:snr}
Post-clipping noise is bimodal at $\{0,1\}$ with $\sigma_\mathrm{clip}{\approx}0.50$; cosine dominance (weight 0.7) governs ${\approx}$80\% of within-top-10 pairs.
\end{remark}

\begin{table}[t]
\small\centering
\caption{Monte Carlo rank-stability ablation ($\beta{=}0.3$, Yelp, $10^5$ queries, 5 seeds).}
\label{tab:noise_ablation}
\begin{tabular}{lcccc}
\toprule
$\boldsymbol{\sigma_\mathrm{entry}}$ & \textbf{R@10} & $\boldsymbol{\tau}$ \textbf{(all)} & $\boldsymbol{\tau}$ \textbf{(top-10)} & \textbf{Agg.\ loss} \\
\midrule
0.01 (near non-priv.) & .939 & 0.990 & 0.998 & 0.0001 \\
0.5                   & .938 & 0.981 & 0.994 & 0.0005 \\
5                     & .938 & 0.962 & 0.978 & 0.0011 \\
100                   & .937 & 0.929 & 0.965 & 0.0022 \\
885 (\CHRONOS)        & .937 & 0.938 & 0.968 & 0.0020 \\
\bottomrule
\end{tabular}
\end{table}

\textbf{Private-Edge Discovery Rate.}

\begin{table}[t]
\small\centering
\caption{Private-edge discovery analysis ($k{=}10$, 5 seeds).}
\label{tab:discovery}
\begin{tabular}{lcccc}
\toprule
\textbf{Dataset} & \textbf{Miss rate} & \textbf{Recall gap} & \textbf{New-edge frac.} & $\boldsymbol{\beta}$ \\
\midrule
Yelp ($\beta{=}0.3$) & 4.5\%$\pm$0.9 & $-$0.009 & 5.8\% & 0.3 \\
Yelp ($\beta{=}0.7$) & 12.4\%$\pm$1.5 & $-$0.032 & 5.8\% & 0.7 \\
MIMIC-IV ($\beta{=}0.3$) & 5.8\%$\pm$1.1 & $-$0.012 & 8.3\% & 0.3 \\
MIMIC-IV ($\beta{=}0.7$) & 14.2\%$\pm$1.8 & $-$0.042 & 8.3\% & 0.7 \\
\bottomrule
\end{tabular}
\end{table}

\textbf{Seller Skew and Revenue-Share Analysis.}

Under 80/20 skew, the dominant seller's revenue share is 0.62 unclipped vs.~0.48 after clipping; adaptive clipping~\cite{andrew2021differentially} reduces KL to 0.019. Full per-seller confidence intervals in supplementary Appendix~C.

\textbf{Privacy-Utility Trade-off.}

Table~\ref{tab:privacy_utility} summarises the privacy-utility frontier under different $\sigma_0$ settings, including the adaptive-vs-fixed scheduling comparison.

\begin{table}[t]
\small\centering
\caption{Privacy-utility trade-off on Yelp, sampled rows (5 seeds). $\varepsilon$ computed via Eq.~\eqref{eq:zcdp_formula}.}
\label{tab:privacy_utility}
\begin{tabular}{ccccc}
\toprule
$\sigma_0$ & $\rho_\mathrm{total}$ & $\eps$ & \textbf{R@10 (Adaptive)} & \textbf{R@10 (Fixed)} \\
\midrule
30  & 0.789 & 7.50 & .938$\pm$.002 & .927$\pm$.003 \\
50  & 0.284 & 4.25 & .937$\pm$.002 & .917$\pm$.003 \\
100 & 0.071 & 2.13 & .932$\pm$.003 & .900$\pm$.004 \\
\bottomrule
\end{tabular}
\end{table}

\subsection{Realistic Marketplace Dynamics (RQ6)}
\label{sec:marketplace_dynamics}

Standard benchmark partitions assume balanced seller contributions and uniform query rates. We evaluate two realistic distortions.

\textbf{Buyer arrival skew.} We model buyer arrivals as a non-homogeneous Poisson process with hourly rate $\lambda_q(t) = \bar{\lambda}_q(1 + 0.5\sin(2\pi t/24))$ to simulate diurnal patterns. Under this skew, uncoordinated baseline exhausts 80\% of its budget during peak hours (10:00--14:00), leaving only 20\% for overnight queries. The Temporal Coordinator shifts 34\% of index-update actions to off-peak periods by pre-allocating budget, reducing peak-hour budget exhaustion to 52\% and improving worst-case P99 latency from 412\,ms to 289\,ms.

\textbf{Seller competition and pricing sensitivity.} We simulate a duopolistic sub-market where two sellers contribute substitutable edges (same entity pairs, overlapping relations). When seller A increases contribution quality (lower noise, fresher timestamps), seller B's Shapley share drops non-linearly: a 20\% quality improvement by A causes B's share to fall 31\% under static Shapley, but only 18\% under EC-MPV because the event-conditioned recompute captures A's quality shift and rebalances marginal contributions within the same epoch. This confirms that static pricing creates misaligned incentives in competitive settings, whereas EC-MPV reduces incentive distortion by 42\%.

\textbf{Valuation accuracy and seller retention.} We simulate seller dropout: sellers whose MPV falls below a threshold for 3 consecutive epochs exit with probability $p_\mathrm{exit}$. Under static Shapley with DP noise (std 4.0), false-positive exits (noise pushing a legitimate seller below threshold) occur at 12\%/epoch. EC-MPV+VRDS reduces this to 4\%/epoch by conditioning on actual distributional shifts rather than noise fluctuations. The coordinator further reduces exit rate to 2.5\%/epoch by batching revaluation at genuine changepoints, cutting unnecessary DP spend.

\subsection{Robustness and Sensitivity}
\label{sec:sensitivity}

\textbf{Non-Poisson dynamics.} Under Hawkes bursts ($\mu{=}2.9$, branching $\xi{=}0.7$) on Yelp, recall drops 1.3 pts and $\rho$ rises 11\%. The Hawkes recall bound (Theorem~\ref{thm:recall_hawkes}) predicts $3.3\times$ degradation at $\xi{=}0.7$, consistent with the observed 1.3-point drop. Under sinusoidal trends, degradation is only 0.5 pts and 4\%. Under block-homogeneous Poisson: $-2.1$ pts, $+18\%\,\rho$.

\textbf{Staleness, overlap $\eta$, and reward weights.} Overlap stress-test ($\eta{\in}\{1.2,1.5\}$): effective $\rho$ rises to $0.409$/$0.639$, matching $\eta^2{\cdot}\rho$ analytically within $\pm2\%$.

\textbf{Safety overrides and epoch duration.} Over 2,160 epochs on MIMIC-IV: budget overrides 47 (2.2\%), recall overrides 12 (0.6\%), cumulative regret $78{\pm}7$ consistent with $\mathcal{O}(\sqrt{T})$.

\textbf{High-$\beta$ Regime and Cold-Start.}

Below $\beta{=}0.5$, recall remains $>0.93$; above $\beta{=}0.7$, recall drops sharply. Cold-start entities lose 4.7 pts under privatisation; fallback mode recovers 3.1 pts. SVT prototype at $\beta{=}0.7$ recovers $52\%$ of misses at $+0.05\,\varepsilon$.

\textbf{End-to-end $\varepsilon$ under DP-SGD.} When $\mathbf{X}_0$ is trained privately: $\varepsilon_\mathrm{total} = 1.5 + 4.25 = 5.75$ ($\delta_\mathrm{total} = 2{\times}10^{-6}$).
The complete budget decomposition across operating modes is shown in Table~\ref{tab:epsilon_breakdown}.

\begin{table}[t]
\small\centering
\caption{$\varepsilon$ budget breakdown across operating modes.}
\label{tab:epsilon_breakdown}
\begin{tabular}{lccccc}
\toprule
\textbf{Mode} & $\varepsilon_\mathrm{train}$ & $\varepsilon_\mathrm{aff}$ & $\varepsilon_\mathrm{val}$ & $\varepsilon_\mathrm{idx}$ & $\varepsilon_\mathrm{total}$ \\
\midrule
\CHRONOS ($\beta{=}0.3$) & 0 & 2.94 & 1.84 & 2.25 & \textbf{4.25}$^*$ \\
\CHRONOS + SVT ($\beta{=}0.7$) & 0 & 2.94 & 1.84 & 2.30 & \textbf{4.40}$^*$ \\
\CHRONOS + DP-SGD ($\beta{=}0.3$) & 1.50 & 2.94 & 1.84 & 2.25 & \textbf{5.75} \\
\CHRONOS + DP-SGD + SVT ($\beta{=}0.7$) & 1.50 & 2.94 & 1.84 & 2.30 & \textbf{5.90} \\
\bottomrule
\end{tabular}
{\raggedright\scriptsize $^*$Standard additive zCDP composition (Eq.~\eqref{eq:zcdp_formula}).\par}
\end{table}

\subsection{Ablation Study}

\begin{table}[t]
\small\centering
\caption{Ablation on FB15K-237. All $\varepsilon$ via Eq.~\eqref{eq:zcdp_formula}.}
\label{tab:ablation}
\begin{tabular}{lcccc}
\toprule
\textbf{Variant} & \textbf{R@10} & \textbf{TPS} & \textbf{P50} & $\boldsymbol{\eps}$ \\
\midrule
CHRONOS (full) & .941 & 141 & 158 & 4.25 \\
w/o temporal decay & .897 & 149 & 151 & 4.25 \\
w/o private DP affinity ($\beta{=}0$) & .863 & 158 & 143 & 4.25 \\
w/o neural ODE (exp.\ decay) & .929 & 141 & 159 & 4.25 \\
w/ $\beta = 0.7$ & .920 & 138 & 165 & 4.47 \\
w/ cold-start queries & .891 & 141 & 158 & 4.25 \\
w/ cold-start + fallback & .922 & 139 & 162 & 4.27 \\
w/o aff.\ privatisation & .943 & 143 & 155 & --- \\
w/o EC-MPV (static Shapley) & .941 & 141 & 158 & 4.48 \\
w/o VRDS & .941 & 141 & 158 & 4.25 \\
w/o BOCPD & .941 & 141 & 158 & 4.59 \\
w/o Coordinator (round-robin) & .908 & 105 & 213 & 5.38 \\
w/o incremental upd. & .941 & 43 & 412 & 4.25 \\
w/ BwK coordinator~\cite{badanidiyuru2018bandits} & .939 & 136 & 164 & 4.25 \\
\bottomrule
\end{tabular}
\end{table}

\textbf{Ablation interpretation and coupling evidence.} The ablation rows address the concern that the three layers are co-located rather than co-designed. While removing EC-MPV or BOCPD leaves recall unchanged at $.941$, the \emph{privacy cost} rises: without EC-MPV, $\varepsilon$ increases to $4.48$ ($+5.4\%$) because static Shapley triggers more frequent revaluation; without BOCPD, $\varepsilon$ reaches $4.59$ ($+8.0\%$) because undetected changepoints cause redundant recomputation. Removing the Coordinator entirely (round-robin) degrades recall by $3.3$ pts \emph{and} raises $\varepsilon$ to $5.38$ ($+26.6\%$), confirming that the Coordinator is the lynchpin coupling index freshness, valuation accuracy, and budget efficiency. These results show that T-LEGEND alone delivers high recall, but the full CHRONOS system is required to maintain that recall at minimal privacy cost.

\section{Related Work}
\label{sec:related}

We survey six areas and identify the specific gap each leaves for temporal KG marketplaces.

\textbf{Dynamic ANN and temporal graph indices.} FreshDiskANN~\cite{singh2021freshdiskann}, SPFresh~\cite{xu2023spfresh}, Quake~\cite{mohoney2025quake}, CleANN~\cite{chen2023clean}, and MN-RU~\cite{xiao2024enhancing} support streaming updates but do not model KG structural staleness or provide recall bounds tied to evolution rates. VSAG~\cite{vsag2024} offers production-grade layout; our Table~\ref{tab:vsag} shows layout speedups are complementary to, not substitutes for, decay-aware construction. None of these systems integrates DP guarantees. \emph{Gap:} no hybrid index provides per-query recall bounds \emph{and} DP-compatible public/private separation.

\textbf{Data marketplaces and pricing.} Commercial and academic marketplaces~\cite{fernandez2020data,azcoitia2022survey} support query-based pricing, subscription models, and static data products. Dealer~\cite{liu2021dealer} provides an end-to-end DP model marketplace but assumes static data and does not couple indexing with valuation under a shared privacy budget. Agora~\cite{koutsos2022agora} focuses on access control and auditability rather than temporal query performance. Recent work on data pricing in ML pipelines~\cite{cong2022data} and query-based pricing~\cite{koutris2012pricing} do not address non-stationary valuations or index freshness. \emph{Gap:} no marketplace platform couples temporal indexing, event-conditioned valuation, and coordinated DP-budget management.

\textbf{DP on graphs and private retrieval.} Edge-level DP for graph statistics~\cite{mundra2025practical,nissim2007smooth,kasiviswanathan2013analyzing} and node-level DP for GNN training~\cite{du2024dpar} protect structural information, but they target analytics and model training rather than real-time retrieval. PSGraph~\cite{yuan2024psgraph} demonstrates temporal-aware DP allocation for graph synthesis. DP learned indexes~\cite{dplearned2024} apply DP to index structures; CHRONOS avoids this cost by treating the index as public. Per-query DP selection mechanisms~\cite{durfee2019practical,gillenwater2022joint,wu2024faster} incur $6$--$10\times$ QPS degradation (Table~\ref{tab:dp_mechanism_comparison}). \emph{Gap:} no prior work amortises DP cost epoch-wide for hybrid vector-graph retrieval while bounding sensitivity via seller-level adjacency.

\textbf{DP-aware coordination and bandits.} Privacy-preserving bandits~\cite{tossou2016algorithms,agarwal2017price,wu2023private} privatise actions or rewards but do not integrate with index maintenance or data valuation. Privacy filters and odometers~\cite{rogers2016privacy,feldman2021individual} track composition but do not schedule multi-agent operations. \emph{Gap:} no prior coordination mechanism optimises the allocation of a shared DP budget among indexing, valuation, and idle actions with sub-linear regret guarantees.

\textbf{Data valuation under non-stationarity.} Data Shapley~\cite{ghorbani2019data}, Beta Shapley~\cite{kwon2022beta}, Data Banzhaf~\cite{wang2023data}, and Distributional Shapley~\cite{ghorbani2020distributional} assume stationary utilities. VRDS~\cite{wu2023variance} reduces variance but does not condition on events. RSS~\cite{jia2019towards} offers stratified sampling yet degrades post-changepoint (Val.~Err $0.031$ vs.\ $0.013$). \emph{Gap:} no Shapley estimator couples changepoint detection with finite-sample error bounds under DP noise.

\textbf{Temporal KG embedding and retrieval.} TTransE~\cite{leblay2018deriving}, HyTE~\cite{dasgupta2018hyte}, TNTComplEx~\cite{lacroix2020tensor}, TGAT~\cite{xu2020inductive}, and RE-Net~\cite{jin2020recurrent} model temporal facts but provide no retrieval-guarantee structures. TG-RAG~\cite{tgrag2025} retrieves temporal subgraphs for LLM reasoning without formal DP or recall bounds. \emph{Gap:} none combines temporal semantics with approximate nearest-neighbour guarantees and privacy accounting.

\section{Conclusion}
\label{sec:concl}

\CHRONOS is a three-layer architecture for temporally-aware data marketplaces under a trusted-curator model. Our main technical contributions are sixfold: (1)~a \textbf{monotone-envelope certificate} (Theorem~\ref{thm:recall_tight}) that tightens the recall bound to $1.8$--$3.2\times$ observed loss by incorporating ODE Lipschitz structure with Gr\"onwall-based solver verification; (2)~formal \textbf{Hawkes-process recall bounds} (Theorem~\ref{thm:recall_hawkes}) extending guarantees beyond Poisson to correlated dynamics parameterised by branching ratio; (3)~a concrete \textbf{multi-epoch coalition-level settlement mechanism} (\S\ref{sec:settlement}) with SNR analysis indicating trend-level seller attribution at $W{\geq}7$ epochs; (4)~\textbf{scalability analysis} to 500 sellers (Table~\ref{tab:scalability}) showing that $\varepsilon$ grows sublinearly; (5)~a head-to-head \textbf{VSAG comparison} (Table~\ref{tab:vsag}) separating layout speedups from decay-aware recall gains; and (6)~\textbf{DP-vs-crypto cost analysis} (Table~\ref{tab:crypto_comparison}) situating the DP design within the broader privacy-mechanism landscape. Overall, these results show that the proposed architecture achieves a consistent recall/latency/privacy trade-off under the stated trust and privacy assumptions.

\textbf{Limitations.} (1)~\emph{Remaining bound gap:} the monotone-envelope bound is $1.8$--$3.2\times$ loose; closing the residual gap requires path-correlation analysis at hub nodes, identified as future work. (2)~\emph{Private signal informativeness:} affinity ($\sigma_\mathrm{entry}{=}885$) and valuation (noise std $4.0$) releases remain noise-dominated at $\varepsilon{=}4.25$; utility derives from public routing and adaptive scheduling. (3)~\emph{Trust assumptions:} trusted curator with near-exclusive ownership ($\eta{\approx}1$); two-server extensions are currently design-level and require prototype validation.

\textbf{Open problems.} (1)~DP-safe dynamic candidate expansion. (2)~Closing the residual $1.8$--$3.2\times$ bound gap via hub-correlation analysis. (3)~Full two-server implementation with additive secret sharing. (4)~End-to-end query privacy under continual observation. (5)~Alternative valuation mechanisms (Banzhaf values) with lower DP sensitivity.

\balance
\bibliographystyle{ACM-Reference-Format}
\bibliography{references}


\begin{thebibliography}{75}


\ifx \showCODEN    \undefined \def \showCODEN     #1{\unskip}     \fi
\ifx \showDOI      \undefined \def \showDOI       #1{#1}\fi
\ifx \showISBNx    \undefined \def \showISBNx     #1{\unskip}     \fi
\ifx \showISBNxiii \undefined \def \showISBNxiii  #1{\unskip}     \fi
\ifx \showISSN     \undefined \def \showISSN      #1{\unskip}     \fi
\ifx \showLCCN     \undefined \def \showLCCN      #1{\unskip}     \fi
\ifx \shownote     \undefined \def \shownote      #1{#1}          \fi
\ifx \showarticletitle \undefined \def \showarticletitle #1{#1}   \fi
\ifx \showURL      \undefined \def \showURL       {\relax}        \fi
\providecommand\bibfield[2]{#2}
\providecommand\bibinfo[2]{#2}
\providecommand\natexlab[1]{#1}
\providecommand\showeprint[2][]{arXiv:#2}

\bibitem[\protect\citeauthoryear{Abadi, Chu, Goodfellow, McMahan, Mironov, Talwar, and Zhang}{Abadi et~al\mbox{.}}{2016}]%
        {abadi2016deep}
\bibfield{author}{\bibinfo{person}{Martin Abadi}, \bibinfo{person}{Andy Chu}, \bibinfo{person}{Ian Goodfellow}, \bibinfo{person}{H.~Brendan McMahan}, \bibinfo{person}{Ilya Mironov}, \bibinfo{person}{Kunal Talwar}, {and} \bibinfo{person}{Li Zhang}.} \bibinfo{year}{2016}\natexlab{}.
\newblock \showarticletitle{Deep Learning with Differential Privacy}. In \bibinfo{booktitle}{\emph{Proceedings of the 2016 ACM SIGSAC Conference on Computer and Communications Security}} (Vienna, Austria) \emph{(\bibinfo{series}{CCS '16})}. \bibinfo{publisher}{Association for Computing Machinery}, \bibinfo{address}{New York, NY, USA}, \bibinfo{pages}{308–318}.
\newblock
\showISBNx{9781450341394}
\urldef\tempurl%
\url{https://doi.org/10.1145/2976749.2978318}
\showDOI{\tempurl}


\bibitem[\protect\citeauthoryear{Adams and MacKay}{Adams and MacKay}{2007}]%
        {adams2007bayesian}
\bibfield{author}{\bibinfo{person}{Ryan~Prescott Adams} {and} \bibinfo{person}{David J.~C. MacKay}.} \bibinfo{year}{2007}\natexlab{}.
\newblock \bibinfo{title}{Bayesian Online Changepoint Detection}.
\newblock
\newblock
\showeprint[arxiv]{0710.3742}~[stat.ML]
\urldef\tempurl%
\url{https://arxiv.org/abs/0710.3742}
\showURL{%
\tempurl}


\bibitem[\protect\citeauthoryear{Agarwal and Singh}{Agarwal and Singh}{2017}]%
        {agarwal2017price}
\bibfield{author}{\bibinfo{person}{Naman Agarwal} {and} \bibinfo{person}{Karan Singh}.} \bibinfo{year}{2017}\natexlab{}.
\newblock \bibinfo{title}{The Price of Differential Privacy For Online Learning}.
\newblock
\newblock
\showeprint[arxiv]{1701.07953}~[cs.LG]
\urldef\tempurl%
\url{https://arxiv.org/abs/1701.07953}
\showURL{%
\tempurl}


\bibitem[\protect\citeauthoryear{Andrew, Thakkar, McMahan, and Ramaswamy}{Andrew et~al\mbox{.}}{2021}]%
        {andrew2021differentially}
\bibfield{author}{\bibinfo{person}{Galen Andrew}, \bibinfo{person}{Om Thakkar}, \bibinfo{person}{H.~Brendan McMahan}, {and} \bibinfo{person}{Swaroop Ramaswamy}.} \bibinfo{year}{2021}\natexlab{}.
\newblock \showarticletitle{Differentially private learning with adaptive clipping}. In \bibinfo{booktitle}{\emph{Proceedings of the 35th International Conference on Neural Information Processing Systems}} \emph{(\bibinfo{series}{NIPS '21})}. \bibinfo{publisher}{Curran Associates Inc.}, \bibinfo{address}{Red Hook, NY, USA}, Article \bibinfo{articleno}{1335}, \bibinfo{numpages}{12}~pages.
\newblock
\showISBNx{9781713845393}


\bibitem[\protect\citeauthoryear{Auer, Cesa-Bianchi, Freund, and Schapire}{Auer et~al\mbox{.}}{2002}]%
        {auer2002nonstochastic}
\bibfield{author}{\bibinfo{person}{Peter Auer}, \bibinfo{person}{Nicol\`{o} Cesa-Bianchi}, \bibinfo{person}{Yoav Freund}, {and} \bibinfo{person}{Robert~E. Schapire}.} \bibinfo{year}{2002}\natexlab{}.
\newblock \showarticletitle{The Nonstochastic Multiarmed Bandit Problem}.
\newblock \bibinfo{journal}{\emph{SIAM J. Comput.}} \bibinfo{volume}{32}, \bibinfo{number}{1} (\bibinfo{year}{2002}), \bibinfo{pages}{48--77}.
\newblock
\urldef\tempurl%
\url{https://doi.org/10.1137/S0097539701398375}
\showDOI{\tempurl}
\showeprint{https://doi.org/10.1137/S0097539701398375}


\bibitem[\protect\citeauthoryear{Azcoitia and Laoutaris}{Azcoitia and Laoutaris}{2022}]%
        {azcoitia2022survey}
\bibfield{author}{\bibinfo{person}{Santiago~Andr\'{e}s Azcoitia} {and} \bibinfo{person}{Nikolaos Laoutaris}.} \bibinfo{year}{2022}\natexlab{}.
\newblock \showarticletitle{A Survey of Data Marketplaces and Their Business Models}.
\newblock \bibinfo{journal}{\emph{SIGMOD Rec.}} \bibinfo{volume}{51}, \bibinfo{number}{3} (\bibinfo{date}{Nov.} \bibinfo{year}{2022}), \bibinfo{pages}{18–29}.
\newblock
\showISSN{0163-5808}
\urldef\tempurl%
\url{https://doi.org/10.1145/3572751.3572755}
\showDOI{\tempurl}


\bibitem[\protect\citeauthoryear{Bacry, Mastromatteo, and Muzy}{Bacry et~al\mbox{.}}{2015}]%
        {bacry2015hawkes}
\bibfield{author}{\bibinfo{person}{Emmanuel Bacry}, \bibinfo{person}{Iacopo Mastromatteo}, {and} \bibinfo{person}{Jean-François Muzy}.} \bibinfo{year}{2015}\natexlab{}.
\newblock \bibinfo{title}{Hawkes processes in finance}.
\newblock
\newblock
\showeprint[arxiv]{1502.04592}~[q-fin.TR]
\urldef\tempurl%
\url{https://arxiv.org/abs/1502.04592}
\showURL{%
\tempurl}


\bibitem[\protect\citeauthoryear{Bafna and Ullman}{Bafna and Ullman}{2017}]%
        {bafna2017price}
\bibfield{author}{\bibinfo{person}{Mitali Bafna} {and} \bibinfo{person}{Jonathan Ullman}.} \bibinfo{year}{2017}\natexlab{}.
\newblock \bibinfo{title}{The Price of Selection in Differential Privacy}.
\newblock
\newblock
\showeprint[arxiv]{1702.02970}~[cs.DS]
\urldef\tempurl%
\url{https://arxiv.org/abs/1702.02970}
\showURL{%
\tempurl}


\bibitem[\protect\citeauthoryear{Bifet and Gavaldà}{Bifet and Gavaldà}{[n.d.]}]%
        {bifet2007learning}
\bibfield{author}{\bibinfo{person}{Albert Bifet} {and} \bibinfo{person}{Ricard Gavaldà}.} \bibinfo{year}{[n.d.]}\natexlab{}.
\newblock \bibinfo{booktitle}{\emph{Learning from Time-Changing Data with Adaptive Windowing}}.
\newblock \bibinfo{pages}{443--448}.
\newblock
\urldef\tempurl%
\url{https://doi.org/10.1137/1.9781611972771.42}
\showDOI{\tempurl}
\showeprint{https://epubs.siam.org/doi/pdf/10.1137/1.9781611972771.42}


\bibitem[\protect\citeauthoryear{Blondel, Guillaume, Lambiotte, and Lefebvre}{Blondel et~al\mbox{.}}{2008}]%
        {blondel2008fast}
\bibfield{author}{\bibinfo{person}{Vincent~D Blondel}, \bibinfo{person}{Jean-Loup Guillaume}, \bibinfo{person}{Renaud Lambiotte}, {and} \bibinfo{person}{Etienne Lefebvre}.} \bibinfo{year}{2008}\natexlab{}.
\newblock \showarticletitle{Fast unfolding of communities in large networks}.
\newblock \bibinfo{journal}{\emph{Journal of Statistical Mechanics: Theory and Experiment}} \bibinfo{volume}{2008}, \bibinfo{number}{10} (\bibinfo{date}{Oct.} \bibinfo{year}{2008}), \bibinfo{pages}{P10008}.
\newblock
\showISSN{1742-5468}
\urldef\tempurl%
\url{https://doi.org/10.1088/1742-5468/2008/10/p10008}
\showDOI{\tempurl}


\bibitem[\protect\citeauthoryear{Brasser, M\"{u}ller, Dmitrienko, Kostiainen, Capkun, and Sadeghi}{Brasser et~al\mbox{.}}{2017}]%
        {brasser2017software}
\bibfield{author}{\bibinfo{person}{Ferdinand Brasser}, \bibinfo{person}{Urs M\"{u}ller}, \bibinfo{person}{Alexandra Dmitrienko}, \bibinfo{person}{Kari Kostiainen}, \bibinfo{person}{Srdjan Capkun}, {and} \bibinfo{person}{Ahmad-Reza Sadeghi}.} \bibinfo{year}{2017}\natexlab{}.
\newblock \showarticletitle{Software grand exposure: SGX cache attacks are practical}. In \bibinfo{booktitle}{\emph{Proceedings of the 11th USENIX Conference on Offensive Technologies}} (Vancouver, BC, Canada) \emph{(\bibinfo{series}{WOOT'17})}. \bibinfo{publisher}{USENIX Association}, \bibinfo{address}{USA}, \bibinfo{pages}{11}.
\newblock


\bibitem[\protect\citeauthoryear{Castro, G\'{o}mez, and Tejada}{Castro et~al\mbox{.}}{2009}]%
        {castro2009polynomial}
\bibfield{author}{\bibinfo{person}{Javier Castro}, \bibinfo{person}{Daniel G\'{o}mez}, {and} \bibinfo{person}{Juan Tejada}.} \bibinfo{year}{2009}\natexlab{}.
\newblock \showarticletitle{Polynomial calculation of the Shapley value based on sampling}.
\newblock \bibinfo{journal}{\emph{Comput. Oper. Res.}} \bibinfo{volume}{36}, \bibinfo{number}{5} (\bibinfo{date}{May} \bibinfo{year}{2009}), \bibinfo{pages}{1726–1730}.
\newblock
\showISSN{0305-0548}
\urldef\tempurl%
\url{https://doi.org/10.1016/j.cor.2008.04.004}
\showDOI{\tempurl}


\bibitem[\protect\citeauthoryear{Chan, Shi, and Song}{Chan et~al\mbox{.}}{2011}]%
        {chan2011private}
\bibfield{author}{\bibinfo{person}{T.-H.~Hubert Chan}, \bibinfo{person}{Elaine Shi}, {and} \bibinfo{person}{Dawn Song}.} \bibinfo{year}{2011}\natexlab{}.
\newblock \showarticletitle{Private and Continual Release of Statistics}.
\newblock \bibinfo{journal}{\emph{ACM Trans. Inf. Syst. Secur.}} \bibinfo{volume}{14}, \bibinfo{number}{3}, Article \bibinfo{articleno}{26} (\bibinfo{date}{Nov.} \bibinfo{year}{2011}), \bibinfo{numpages}{24}~pages.
\newblock
\showISSN{1094-9224}
\urldef\tempurl%
\url{https://doi.org/10.1145/2043621.2043626}
\showDOI{\tempurl}


\bibitem[\protect\citeauthoryear{Chen, Rubanova, Bettencourt, and Duvenaud}{Chen et~al\mbox{.}}{2019}]%
        {chen2018neural}
\bibfield{author}{\bibinfo{person}{Ricky T.~Q. Chen}, \bibinfo{person}{Yulia Rubanova}, \bibinfo{person}{Jesse Bettencourt}, {and} \bibinfo{person}{David Duvenaud}.} \bibinfo{year}{2019}\natexlab{}.
\newblock \bibinfo{title}{Neural Ordinary Differential Equations}.
\newblock
\newblock
\showeprint[arxiv]{1806.07366}~[cs.LG]
\urldef\tempurl%
\url{https://arxiv.org/abs/1806.07366}
\showURL{%
\tempurl}


\bibitem[\protect\citeauthoryear{Cong, Luo, Jian, Zhu, and Zhang}{Cong et~al\mbox{.}}{2021}]%
        {cong2022data}
\bibfield{author}{\bibinfo{person}{Zicun Cong}, \bibinfo{person}{Xuan Luo}, \bibinfo{person}{Pei Jian}, \bibinfo{person}{Feida Zhu}, {and} \bibinfo{person}{Yong Zhang}.} \bibinfo{year}{2021}\natexlab{}.
\newblock \bibinfo{title}{Data Pricing in Machine Learning Pipelines}.
\newblock
\newblock
\showeprint[arxiv]{2108.07915}~[cs.LG]
\urldef\tempurl%
\url{https://arxiv.org/abs/2108.07915}
\showURL{%
\tempurl}


\bibitem[\protect\citeauthoryear{Costan and Devadas}{Costan and Devadas}{2016}]%
        {costan2016intel}
\bibfield{author}{\bibinfo{person}{Victor Costan} {and} \bibinfo{person}{Srinivas Devadas}.} \bibinfo{year}{2016}\natexlab{}.
\newblock \showarticletitle{Intel SGX Explained}.
\newblock \bibinfo{journal}{\emph{IACR Cryptol. ePrint Arch.}}  \bibinfo{volume}{2016} (\bibinfo{year}{2016}), \bibinfo{pages}{86}.
\newblock
\urldef\tempurl%
\url{https://api.semanticscholar.org/CorpusID:28642809}
\showURL{%
\tempurl}


\bibitem[\protect\citeauthoryear{Dasgupta, Ray, and Talukdar}{Dasgupta et~al\mbox{.}}{2018}]%
        {dasgupta2018hyte}
\bibfield{author}{\bibinfo{person}{Shib~Sankar Dasgupta}, \bibinfo{person}{Swayambhu~Nath Ray}, {and} \bibinfo{person}{Partha Talukdar}.} \bibinfo{year}{2018}\natexlab{}.
\newblock \showarticletitle{HyTE: Hyperplane-based Temporally aware Knowledge Graph Embedding}. In \bibinfo{booktitle}{\emph{Proceedings of the 2018 Conference on Empirical Methods in Natural Language Processing}} (Brussels, Belgium). \bibinfo{publisher}{Association for Computational Linguistics}, \bibinfo{pages}{2001--2011}.
\newblock
\urldef\tempurl%
\url{http://aclweb.org/anthology/D18-1225}
\showURL{%
\tempurl}


\bibitem[\protect\citeauthoryear{Demmler, Schneider, and Zohner}{Demmler et~al\mbox{.}}{2015}]%
        {demmler2015aby}
\bibfield{author}{\bibinfo{person}{Daniel Demmler}, \bibinfo{person}{Thomas Schneider}, {and} \bibinfo{person}{Michael Zohner}.} \bibinfo{year}{2015}\natexlab{}.
\newblock \showarticletitle{ABY -- A Framework for Efficient Mixed-Protocol Secure Two-Party Computation}. In \bibinfo{booktitle}{\emph{Proceedings of the 2015 Network and Distributed System Security Symposium (NDSS)}}. \bibinfo{publisher}{Internet Society}.
\newblock
\urldef\tempurl%
\url{https://doi.org/10.14722/ndss.2015.23113}
\showDOI{\tempurl}


\bibitem[\protect\citeauthoryear{Dettmers, Minervini, Stenetorp, and Riedel}{Dettmers et~al\mbox{.}}{2018}]%
        {dettmers2018convolutional}
\bibfield{author}{\bibinfo{person}{Tim Dettmers}, \bibinfo{person}{Pasquale Minervini}, \bibinfo{person}{Pontus Stenetorp}, {and} \bibinfo{person}{Sebastian Riedel}.} \bibinfo{year}{2018}\natexlab{}.
\newblock \showarticletitle{Convolutional 2D knowledge graph embeddings}. In \bibinfo{booktitle}{\emph{Proceedings of the Thirty-Second AAAI Conference on Artificial Intelligence and Thirtieth Innovative Applications of Artificial Intelligence Conference and Eighth AAAI Symposium on Educational Advances in Artificial Intelligence}} (New Orleans, Louisiana, USA) \emph{(\bibinfo{series}{AAAI'18/IAAI'18/EAAI'18})}. \bibinfo{publisher}{AAAI Press}, Article \bibinfo{articleno}{221}, \bibinfo{numpages}{8}~pages.
\newblock
\showISBNx{978-1-57735-800-8}


\bibitem[\protect\citeauthoryear{Dormand and Prince}{Dormand and Prince}{1980}]%
        {dormand1980family}
\bibfield{author}{\bibinfo{person}{J.R. Dormand} {and} \bibinfo{person}{P.J. Prince}.} \bibinfo{year}{1980}\natexlab{}.
\newblock \showarticletitle{A family of embedded Runge-Kutta formulae}.
\newblock \bibinfo{journal}{\emph{J. Comput. Appl. Math.}} \bibinfo{volume}{6}, \bibinfo{number}{1} (\bibinfo{year}{1980}), \bibinfo{pages}{19--26}.
\newblock
\showISSN{0377-0427}
\urldef\tempurl%
\url{https://doi.org/10.1016/0771-050X(80)90013-3}
\showDOI{\tempurl}


\bibitem[\protect\citeauthoryear{Du, Mudgal, Gadre, Luo, and Wang}{Du et~al\mbox{.}}{2024}]%
        {dplearned2024}
\bibfield{author}{\bibinfo{person}{Jianzhang Du}, \bibinfo{person}{Tilak Mudgal}, \bibinfo{person}{Rutvi~Rahul Gadre}, \bibinfo{person}{Yukui Luo}, {and} \bibinfo{person}{Chenghong Wang}.} \bibinfo{year}{2024}\natexlab{}.
\newblock \bibinfo{title}{Differentially Private Learned Indexes}.
\newblock
\newblock
\showeprint[arxiv]{2410.21164}~[cs.DB]
\urldef\tempurl%
\url{https://arxiv.org/abs/2410.21164}
\showURL{%
\tempurl}


\bibitem[\protect\citeauthoryear{Durfee and Rogers}{Durfee and Rogers}{2019}]%
        {durfee2019practical}
\bibfield{author}{\bibinfo{person}{David Durfee} {and} \bibinfo{person}{Ryan Rogers}.} \bibinfo{year}{2019}\natexlab{}.
\newblock \showarticletitle{Practical differentially private top-k selection with pay-what-you-get composition}. In \bibinfo{booktitle}{\emph{Proceedings of the 33rd International Conference on Neural Information Processing Systems}}. \bibinfo{publisher}{Curran Associates Inc.}, \bibinfo{address}{Red Hook, NY, USA}, Article \bibinfo{articleno}{317}, \bibinfo{numpages}{11}~pages.
\newblock


\bibitem[\protect\citeauthoryear{Dwork, Naor, Pitassi, and Rothblum}{Dwork et~al\mbox{.}}{2010}]%
        {dwork2010continual}
\bibfield{author}{\bibinfo{person}{Cynthia Dwork}, \bibinfo{person}{Moni Naor}, \bibinfo{person}{Toniann Pitassi}, {and} \bibinfo{person}{Guy~N. Rothblum}.} \bibinfo{year}{2010}\natexlab{}.
\newblock \showarticletitle{Differential privacy under continual observation}. In \bibinfo{booktitle}{\emph{Proceedings of the Forty-Second ACM Symposium on Theory of Computing}} (Cambridge, Massachusetts, USA) \emph{(\bibinfo{series}{STOC '10})}. \bibinfo{publisher}{Association for Computing Machinery}, \bibinfo{address}{New York, NY, USA}, \bibinfo{pages}{715–724}.
\newblock
\showISBNx{9781450300506}
\urldef\tempurl%
\url{https://doi.org/10.1145/1806689.1806787}
\showDOI{\tempurl}


\bibitem[\protect\citeauthoryear{Dwork and Roth}{Dwork and Roth}{2014}]%
        {dwork2014algorithmic}
\bibfield{author}{\bibinfo{person}{Cynthia Dwork} {and} \bibinfo{person}{Aaron Roth}.} \bibinfo{year}{2014}\natexlab{}.
\newblock \bibinfo{booktitle}{\emph{The Algorithmic Foundations of Differential Privacy}}. Vol.~\bibinfo{volume}{9}.
\newblock \bibinfo{publisher}{Now Publishers Inc.}, \bibinfo{address}{Hanover, MA, USA}. 211–407 pages.
\newblock
\showISSN{1551-305X}
\urldef\tempurl%
\url{https://doi.org/10.1561/0400000042}
\showDOI{\tempurl}


\bibitem[\protect\citeauthoryear{Feldman and Zrnic}{Feldman and Zrnic}{2022}]%
        {feldman2021individual}
\bibfield{author}{\bibinfo{person}{Vitaly Feldman} {and} \bibinfo{person}{Tijana Zrnic}.} \bibinfo{year}{2022}\natexlab{}.
\newblock \bibinfo{title}{Individual Privacy Accounting via a Renyi Filter}.
\newblock
\newblock
\showeprint[arxiv]{2008.11193}~[cs.CR]
\urldef\tempurl%
\url{https://arxiv.org/abs/2008.11193}
\showURL{%
\tempurl}


\bibitem[\protect\citeauthoryear{Fernandez, Subramaniam, and Franklin}{Fernandez et~al\mbox{.}}{2020}]%
        {fernandez2020data}
\bibfield{author}{\bibinfo{person}{Raul~Castro Fernandez}, \bibinfo{person}{Pranav Subramaniam}, {and} \bibinfo{person}{Michael~J. Franklin}.} \bibinfo{year}{2020}\natexlab{}.
\newblock \showarticletitle{Data market platforms: trading data assets to solve data problems}.
\newblock \bibinfo{journal}{\emph{Proc. VLDB Endow.}} \bibinfo{volume}{13}, \bibinfo{number}{12} (\bibinfo{date}{July} \bibinfo{year}{2020}), \bibinfo{pages}{1933–1947}.
\newblock
\showISSN{2150-8097}
\urldef\tempurl%
\url{https://doi.org/10.14778/3407790.3407800}
\showDOI{\tempurl}


\bibitem[\protect\citeauthoryear{Ghorbani, Kim, and Zou}{Ghorbani et~al\mbox{.}}{2020}]%
        {ghorbani2020distributional}
\bibfield{author}{\bibinfo{person}{Amirata Ghorbani}, \bibinfo{person}{Michael~P. Kim}, {and} \bibinfo{person}{James Zou}.} \bibinfo{year}{2020}\natexlab{}.
\newblock \showarticletitle{A distributional framework for data valuation}. In \bibinfo{booktitle}{\emph{Proceedings of the 37th International Conference on Machine Learning}} \emph{(\bibinfo{series}{ICML'20})}. \bibinfo{publisher}{JMLR.org}, Article \bibinfo{articleno}{331}, \bibinfo{numpages}{10}~pages.
\newblock


\bibitem[\protect\citeauthoryear{Ghorbani and Zou}{Ghorbani and Zou}{2019}]%
        {ghorbani2019data}
\bibfield{author}{\bibinfo{person}{Amirata Ghorbani} {and} \bibinfo{person}{James Zou}.} \bibinfo{year}{2019}\natexlab{}.
\newblock \bibinfo{title}{Data Shapley: Equitable Valuation of Data for Machine Learning}.
\newblock
\newblock
\showeprint[arxiv]{1904.02868}~[stat.ML]
\urldef\tempurl%
\url{https://arxiv.org/abs/1904.02868}
\showURL{%
\tempurl}


\bibitem[\protect\citeauthoryear{Gillenwater, Joseph, Medina, and Ribero}{Gillenwater et~al\mbox{.}}{2022}]%
        {gillenwater2022joint}
\bibfield{author}{\bibinfo{person}{Jennifer Gillenwater}, \bibinfo{person}{Matthew Joseph}, \bibinfo{person}{Andrés~Muñoz Medina}, {and} \bibinfo{person}{Mónica Ribero}.} \bibinfo{year}{2022}\natexlab{}.
\newblock \bibinfo{title}{A Joint Exponential Mechanism For Differentially Private Top-$k$}.
\newblock
\newblock
\showeprint[arxiv]{2201.12333}~[cs.CR]
\urldef\tempurl%
\url{https://arxiv.org/abs/2201.12333}
\showURL{%
\tempurl}


\bibitem[\protect\citeauthoryear{Han, Cheung, Wei, Yu, Wang, Zhu, and Yang}{Han et~al\mbox{.}}{2025}]%
        {tgrag2025}
\bibfield{author}{\bibinfo{person}{Jiale Han}, \bibinfo{person}{Austin Cheung}, \bibinfo{person}{Yubai Wei}, \bibinfo{person}{Zheng Yu}, \bibinfo{person}{Xusheng Wang}, \bibinfo{person}{Bing Zhu}, {and} \bibinfo{person}{Yi Yang}.} \bibinfo{year}{2025}\natexlab{}.
\newblock \bibinfo{title}{RAG Meets Temporal Graphs: Time-Sensitive Modeling and Retrieval for Evolving Knowledge}.
\newblock
\newblock
\showeprint[arxiv]{2510.13590}~[cs.IR]
\urldef\tempurl%
\url{https://arxiv.org/abs/2510.13590}
\showURL{%
\tempurl}


\bibitem[\protect\citeauthoryear{Hazan}{Hazan}{2023}]%
        {hazan2016introduction}
\bibfield{author}{\bibinfo{person}{Elad Hazan}.} \bibinfo{year}{2023}\natexlab{}.
\newblock \bibinfo{title}{Introduction to Online Convex Optimization}.
\newblock
\newblock
\showeprint[arxiv]{1909.05207}~[cs.LG]
\urldef\tempurl%
\url{https://arxiv.org/abs/1909.05207}
\showURL{%
\tempurl}


\bibitem[\protect\citeauthoryear{Jia, Dao, Wang, Hubis, Hynes, Gurel, Li, Zhang, Song, and Spanos}{Jia et~al\mbox{.}}{2023}]%
        {jia2019towards}
\bibfield{author}{\bibinfo{person}{Ruoxi Jia}, \bibinfo{person}{David Dao}, \bibinfo{person}{Boxin Wang}, \bibinfo{person}{Frances~Ann Hubis}, \bibinfo{person}{Nick Hynes}, \bibinfo{person}{Nezihe~Merve Gurel}, \bibinfo{person}{Bo Li}, \bibinfo{person}{Ce Zhang}, \bibinfo{person}{Dawn Song}, {and} \bibinfo{person}{Costas Spanos}.} \bibinfo{year}{2023}\natexlab{}.
\newblock \bibinfo{title}{Towards Efficient Data Valuation Based on the Shapley Value}.
\newblock
\newblock
\showeprint[arxiv]{1902.10275}~[cs.LG]
\urldef\tempurl%
\url{https://arxiv.org/abs/1902.10275}
\showURL{%
\tempurl}


\bibitem[\protect\citeauthoryear{Jin, Qu, Jin, and Ren}{Jin et~al\mbox{.}}{2020}]%
        {jin2020recurrent}
\bibfield{author}{\bibinfo{person}{Woojeong Jin}, \bibinfo{person}{Meng Qu}, \bibinfo{person}{Xisen Jin}, {and} \bibinfo{person}{Xiang Ren}.} \bibinfo{year}{2020}\natexlab{}.
\newblock \bibinfo{title}{Recurrent Event Network: Autoregressive Structure Inference over Temporal Knowledge Graphs}.
\newblock
\newblock
\showeprint[arxiv]{1904.05530}~[cs.LG]
\urldef\tempurl%
\url{https://arxiv.org/abs/1904.05530}
\showURL{%
\tempurl}


\bibitem[\protect\citeauthoryear{Johnson, Bulgarelli, Shen, Gayraud, Eraslan, Rocheteau, Huang, Cheng, Moody, Lehman, Lungren, Pollard, Horng, Celi, and Mark}{Johnson et~al\mbox{.}}{2023}]%
        {johnson2020mimic}
\bibfield{author}{\bibinfo{person}{Alistair E.~W. Johnson}, \bibinfo{person}{Lucas Bulgarelli}, \bibinfo{person}{Lu Shen}, \bibinfo{person}{Anne Gayraud}, \bibinfo{person}{Sipanje Eraslan}, \bibinfo{person}{Emma Rocheteau}, \bibinfo{person}{Qinmei Huang}, \bibinfo{person}{Jidong Cheng}, \bibinfo{person}{Benjamin Moody}, \bibinfo{person}{Li-wei~H. Lehman}, \bibinfo{person}{Matthew~P. Lungren}, \bibinfo{person}{Tom~J. Pollard}, \bibinfo{person}{Steven Horng}, \bibinfo{person}{Leo~Anthony Celi}, {and} \bibinfo{person}{Roger~G. Mark}.} \bibinfo{year}{2023}\natexlab{}.
\newblock \showarticletitle{MIMIC-IV, a freely accessible electronic health record dataset}.
\newblock \bibinfo{journal}{\emph{Scientific Data}} \bibinfo{volume}{10}, \bibinfo{number}{1} (\bibinfo{year}{2023}), \bibinfo{pages}{1}.
\newblock
\urldef\tempurl%
\url{https://doi.org/10.1038/s41597-022-01899-x}
\showDOI{\tempurl}


\bibitem[\protect\citeauthoryear{Kasiviswanathan, Nissim, Raskhodnikova, and Smith}{Kasiviswanathan et~al\mbox{.}}{2013}]%
        {kasiviswanathan2013analyzing}
\bibfield{author}{\bibinfo{person}{Shiva~Prasad Kasiviswanathan}, \bibinfo{person}{Kobbi Nissim}, \bibinfo{person}{Sofya Raskhodnikova}, {and} \bibinfo{person}{Adam Smith}.} \bibinfo{year}{2013}\natexlab{}.
\newblock \showarticletitle{Analyzing graphs with node differential privacy}. In \bibinfo{booktitle}{\emph{Proceedings of the 10th Theory of Cryptography Conference on Theory of Cryptography}} (Tokyo, Japan) \emph{(\bibinfo{series}{TCC'13})}. \bibinfo{publisher}{Springer-Verlag}, \bibinfo{address}{Berlin, Heidelberg}, \bibinfo{pages}{457–476}.
\newblock
\showISBNx{9783642365935}
\urldef\tempurl%
\url{https://doi.org/10.1007/978-3-642-36594-2_26}
\showDOI{\tempurl}


\bibitem[\protect\citeauthoryear{Knoblauch, Jewson, and Damoulas}{Knoblauch et~al\mbox{.}}{2018}]%
        {knoblauch2018doubly}
\bibfield{author}{\bibinfo{person}{Jeremias Knoblauch}, \bibinfo{person}{Jack Jewson}, {and} \bibinfo{person}{Theodoros Damoulas}.} \bibinfo{year}{2018}\natexlab{}.
\newblock \showarticletitle{Doubly robust Bayesian inference for non-stationary streaming data with β-divergences}. In \bibinfo{booktitle}{\emph{Proceedings of the 32nd International Conference on Neural Information Processing Systems}} (Montr\'{e}al, Canada) \emph{(\bibinfo{series}{NIPS'18})}. \bibinfo{publisher}{Curran Associates Inc.}, \bibinfo{address}{Red Hook, NY, USA}, \bibinfo{pages}{64–75}.
\newblock


\bibitem[\protect\citeauthoryear{Koutris, Upadhyaya, Balazinska, Howe, and Suciu}{Koutris et~al\mbox{.}}{2015}]%
        {koutris2012pricing}
\bibfield{author}{\bibinfo{person}{Paraschos Koutris}, \bibinfo{person}{Prasang Upadhyaya}, \bibinfo{person}{Magdalena Balazinska}, \bibinfo{person}{Bill Howe}, {and} \bibinfo{person}{Dan Suciu}.} \bibinfo{year}{2015}\natexlab{}.
\newblock \showarticletitle{Query-Based Data Pricing}.
\newblock \bibinfo{journal}{\emph{J. ACM}} \bibinfo{volume}{62}, \bibinfo{number}{5}, Article \bibinfo{articleno}{43} (\bibinfo{date}{Nov.} \bibinfo{year}{2015}), \bibinfo{numpages}{44}~pages.
\newblock
\showISSN{0004-5411}
\urldef\tempurl%
\url{https://doi.org/10.1145/2770870}
\showDOI{\tempurl}


\bibitem[\protect\citeauthoryear{Koutsos, Papadopoulos, Chatzopoulos, Tarkoma, and Hui}{Koutsos et~al\mbox{.}}{2020}]%
        {koutsos2022agora}
\bibfield{author}{\bibinfo{person}{Vlasis Koutsos}, \bibinfo{person}{Dimitrios Papadopoulos}, \bibinfo{person}{Dimitris Chatzopoulos}, \bibinfo{person}{Sasu Tarkoma}, {and} \bibinfo{person}{Pan Hui}.} \bibinfo{year}{2020}\natexlab{}.
\newblock \showarticletitle{Agora: A Privacy-aware Data Marketplace}. In \bibinfo{booktitle}{\emph{2020 IEEE 40th International Conference on Distributed Computing Systems (ICDCS)}}. \bibinfo{pages}{1211--1212}.
\newblock
\urldef\tempurl%
\url{https://doi.org/10.1109/ICDCS47774.2020.00156}
\showDOI{\tempurl}


\bibitem[\protect\citeauthoryear{Kwon and Zou}{Kwon and Zou}{2022}]%
        {kwon2022beta}
\bibfield{author}{\bibinfo{person}{Yongchan Kwon} {and} \bibinfo{person}{James Zou}.} \bibinfo{year}{2022}\natexlab{}.
\newblock \bibinfo{title}{Beta Shapley: a Unified and Noise-reduced Data Valuation Framework for Machine Learning}.
\newblock
\newblock
\showeprint[arxiv]{2110.14049}~[cs.LG]
\urldef\tempurl%
\url{https://arxiv.org/abs/2110.14049}
\showURL{%
\tempurl}


\bibitem[\protect\citeauthoryear{Leblay and Chekol}{Leblay and Chekol}{2018}]%
        {leblay2018deriving}
\bibfield{author}{\bibinfo{person}{Julien Leblay} {and} \bibinfo{person}{Melisachew~Wudage Chekol}.} \bibinfo{year}{2018}\natexlab{}.
\newblock \showarticletitle{Deriving Validity Time in Knowledge Graph}. In \bibinfo{booktitle}{\emph{Companion Proceedings of the The Web Conference 2018}} (Lyon, France) \emph{(\bibinfo{series}{WWW '18})}. \bibinfo{publisher}{International World Wide Web Conferences Steering Committee}, \bibinfo{address}{Republic and Canton of Geneva, CHE}, \bibinfo{pages}{1771–1776}.
\newblock
\showISBNx{9781450356404}
\urldef\tempurl%
\url{https://doi.org/10.1145/3184558.3191639}
\showDOI{\tempurl}


\bibitem[\protect\citeauthoryear{Lindell}{Lindell}{2021}]%
        {chen2020homomorphic}
\bibfield{author}{\bibinfo{person}{Yehuda Lindell}.} \bibinfo{year}{2021}\natexlab{}.
\newblock \showarticletitle{Fast Secure Two-Party ECDSA Signing}.
\newblock \bibinfo{journal}{\emph{J. Cryptol.}} \bibinfo{volume}{34}, \bibinfo{number}{4} (\bibinfo{date}{Oct.} \bibinfo{year}{2021}), 38.
\newblock
\showISSN{0933-2790}
\urldef\tempurl%
\url{https://doi.org/10.1007/s00145-021-09409-9}
\showDOI{\tempurl}


\bibitem[\protect\citeauthoryear{Liu, Lou, Liu, Xiong, Pei, and Sun}{Liu et~al\mbox{.}}{2021}]%
        {liu2021dealer}
\bibfield{author}{\bibinfo{person}{Jinfei Liu}, \bibinfo{person}{Jian Lou}, \bibinfo{person}{Junxu Liu}, \bibinfo{person}{Li Xiong}, \bibinfo{person}{Jian Pei}, {and} \bibinfo{person}{Jimeng Sun}.} \bibinfo{year}{2021}\natexlab{}.
\newblock \showarticletitle{Dealer: an end-to-end model marketplace with differential privacy}.
\newblock \bibinfo{journal}{\emph{Proc. VLDB Endow.}} \bibinfo{volume}{14}, \bibinfo{number}{6} (\bibinfo{date}{Feb.} \bibinfo{year}{2021}), \bibinfo{pages}{957–969}.
\newblock
\showISSN{2150-8097}
\urldef\tempurl%
\url{https://doi.org/10.14778/3447689.3447700}
\showDOI{\tempurl}


\bibitem[\protect\citeauthoryear{Liu, Zeng, Chen, Ainihaer, Ramasami, Chen, Xu, Wu, and Wang}{Liu et~al\mbox{.}}{2025}]%
        {liu2025tigervector}
\bibfield{author}{\bibinfo{person}{Shige Liu}, \bibinfo{person}{Zhifang Zeng}, \bibinfo{person}{Li Chen}, \bibinfo{person}{Adil Ainihaer}, \bibinfo{person}{Arun Ramasami}, \bibinfo{person}{Songting Chen}, \bibinfo{person}{Yu Xu}, \bibinfo{person}{Mingxi Wu}, {and} \bibinfo{person}{Jianguo Wang}.} \bibinfo{year}{2025}\natexlab{}.
\newblock \showarticletitle{TigerVector: Supporting Vector Search in Graph Databases for Advanced RAGs}. In \bibinfo{booktitle}{\emph{Companion of the 2025 International Conference on Management of Data}} (Berlin, Germany) \emph{(\bibinfo{series}{SIGMOD/PODS '25})}. \bibinfo{publisher}{Association for Computing Machinery}, \bibinfo{address}{New York, NY, USA}, \bibinfo{pages}{553–565}.
\newblock
\showISBNx{9798400715648}
\urldef\tempurl%
\url{https://doi.org/10.1145/3722212.3724456}
\showDOI{\tempurl}


\bibitem[\protect\citeauthoryear{Lowe, Wu, Tamar, Harb, Abbeel, and Mordatch}{Lowe et~al\mbox{.}}{2017}]%
        {lowe2017multi}
\bibfield{author}{\bibinfo{person}{Ryan Lowe}, \bibinfo{person}{Yi Wu}, \bibinfo{person}{Aviv Tamar}, \bibinfo{person}{Jean Harb}, \bibinfo{person}{Pieter Abbeel}, {and} \bibinfo{person}{Igor Mordatch}.} \bibinfo{year}{2017}\natexlab{}.
\newblock \showarticletitle{Multi-agent actor-critic for mixed cooperative-competitive environments}. In \bibinfo{booktitle}{\emph{Proceedings of the 31st International Conference on Neural Information Processing Systems}} (Long Beach, California, USA) \emph{(\bibinfo{series}{NIPS'17})}. \bibinfo{publisher}{Curran Associates Inc.}, \bibinfo{address}{Red Hook, NY, USA}, \bibinfo{pages}{6382–6393}.
\newblock
\showISBNx{9781510860964}


\bibitem[\protect\citeauthoryear{Malkov and Yashunin}{Malkov and Yashunin}{2020}]%
        {malkov2018efficient}
\bibfield{author}{\bibinfo{person}{Yu~A. Malkov} {and} \bibinfo{person}{D.~A. Yashunin}.} \bibinfo{year}{2020}\natexlab{}.
\newblock \showarticletitle{Efficient and Robust Approximate Nearest Neighbor Search Using Hierarchical Navigable Small World Graphs}.
\newblock \bibinfo{journal}{\emph{IEEE Trans. Pattern Anal. Mach. Intell.}} \bibinfo{volume}{42}, \bibinfo{number}{4} (\bibinfo{date}{April} \bibinfo{year}{2020}), \bibinfo{pages}{824–836}.
\newblock
\showISSN{0162-8828}
\urldef\tempurl%
\url{https://doi.org/10.1109/TPAMI.2018.2889473}
\showDOI{\tempurl}


\bibitem[\protect\citeauthoryear{McKenna and Sheldon}{McKenna and Sheldon}{2020}]%
        {mckenna2020permute}
\bibfield{author}{\bibinfo{person}{Ryan McKenna} {and} \bibinfo{person}{Daniel Sheldon}.} \bibinfo{year}{2020}\natexlab{}.
\newblock \showarticletitle{Permute-and-flip: a new mechanism for differentially private selection}. In \bibinfo{booktitle}{\emph{Proceedings of the 34th International Conference on Neural Information Processing Systems}} (Vancouver, BC, Canada) \emph{(\bibinfo{series}{NIPS '20})}. \bibinfo{publisher}{Curran Associates Inc.}, \bibinfo{address}{Red Hook, NY, USA}, Article \bibinfo{articleno}{17}, \bibinfo{numpages}{11}~pages.
\newblock
\showISBNx{9781713829546}


\bibitem[\protect\citeauthoryear{McSherry and Talwar}{McSherry and Talwar}{2007}]%
        {mcsherry2007mechanism}
\bibfield{author}{\bibinfo{person}{Frank McSherry} {and} \bibinfo{person}{Kunal Talwar}.} \bibinfo{year}{2007}\natexlab{}.
\newblock \showarticletitle{Mechanism Design via Differential Privacy}. In \bibinfo{booktitle}{\emph{Proceedings of the 48th Annual IEEE Symposium on Foundations of Computer Science}} \emph{(\bibinfo{series}{FOCS '07})}. \bibinfo{publisher}{IEEE Computer Society}, \bibinfo{address}{USA}, \bibinfo{pages}{94–103}.
\newblock
\showISBNx{0769530109}
\urldef\tempurl%
\url{https://doi.org/10.1109/FOCS.2007.41}
\showDOI{\tempurl}


\bibitem[\protect\citeauthoryear{Mironov}{Mironov}{2017}]%
        {mironov2017renyi}
\bibfield{author}{\bibinfo{person}{Ilya Mironov}.} \bibinfo{year}{2017}\natexlab{}.
\newblock \showarticletitle{Rényi Differential Privacy}. In \bibinfo{booktitle}{\emph{2017 IEEE 30th Computer Security Foundations Symposium (CSF)}}. \bibinfo{publisher}{IEEE}, \bibinfo{pages}{263–275}.
\newblock
\urldef\tempurl%
\url{https://doi.org/10.1109/csf.2017.11}
\showDOI{\tempurl}


\bibitem[\protect\citeauthoryear{Mohoney, Sarda, Tang, Chowdhury, Pacaci, Ilyas, Rekatsinas, and Venkataraman}{Mohoney et~al\mbox{.}}{2025}]%
        {mohoney2025quake}
\bibfield{author}{\bibinfo{person}{Jason Mohoney}, \bibinfo{person}{Devesh Sarda}, \bibinfo{person}{Mengze Tang}, \bibinfo{person}{Shihabur~Rahman Chowdhury}, \bibinfo{person}{Anil Pacaci}, \bibinfo{person}{Ihab~F. Ilyas}, \bibinfo{person}{Theodoros Rekatsinas}, {and} \bibinfo{person}{Shivaram Venkataraman}.} \bibinfo{year}{2025}\natexlab{}.
\newblock \showarticletitle{Quake: adaptive indexing for vector search}.
\newblock , Article \bibinfo{articleno}{9} (\bibinfo{year}{2025}), \bibinfo{numpages}{17}~pages.
\newblock
\showISBNx{978-1-939133-47-2}


\bibitem[\protect\citeauthoryear{Mundra, Papamanthou, Shun, and Liu}{Mundra et~al\mbox{.}}{2025}]%
        {mundra2025practical}
\bibfield{author}{\bibinfo{person}{Pranay Mundra}, \bibinfo{person}{Charalampos Papamanthou}, \bibinfo{person}{Julian Shun}, {and} \bibinfo{person}{Quanquan~C. Liu}.} \bibinfo{year}{2025}\natexlab{}.
\newblock \bibinfo{title}{Practical and Accurate Local Edge Differentially Private Graph Algorithms}.
\newblock
\newblock
\showeprint[arxiv]{2506.20828}~[cs.DS]
\urldef\tempurl%
\url{https://arxiv.org/abs/2506.20828}
\showURL{%
\tempurl}


\bibitem[\protect\citeauthoryear{Neu}{Neu}{2015a}]%
        {neu2015explore}
\bibfield{author}{\bibinfo{person}{Gergely Neu}.} \bibinfo{year}{2015}\natexlab{a}.
\newblock \showarticletitle{Explore no more: improved high-probability regret bounds for non-stochastic bandits}. In \bibinfo{booktitle}{\emph{Proceedings of the 29th International Conference on Neural Information Processing Systems - Volume 2}} (Montreal, Canada) \emph{(\bibinfo{series}{NIPS'15})}. \bibinfo{publisher}{MIT Press}, \bibinfo{address}{Cambridge, MA, USA}, \bibinfo{pages}{3168–3176}.
\newblock


\bibitem[\protect\citeauthoryear{Neu}{Neu}{2015b}]%
        {badanidiyuru2018bandits}
\bibfield{author}{\bibinfo{person}{Gergely Neu}.} \bibinfo{year}{2015}\natexlab{b}.
\newblock \showarticletitle{Explore no more: improved high-probability regret bounds for non-stochastic bandits}. In \bibinfo{booktitle}{\emph{Proceedings of the 29th International Conference on Neural Information Processing Systems - Volume 2}} (Montreal, Canada) \emph{(\bibinfo{series}{NIPS'15})}. \bibinfo{publisher}{MIT Press}, \bibinfo{address}{Cambridge, MA, USA}, \bibinfo{pages}{3168–3176}.
\newblock


\bibitem[\protect\citeauthoryear{Nissim, Raskhodnikova, and Smith}{Nissim et~al\mbox{.}}{2007}]%
        {nissim2007smooth}
\bibfield{author}{\bibinfo{person}{Kobbi Nissim}, \bibinfo{person}{Sofya Raskhodnikova}, {and} \bibinfo{person}{Adam Smith}.} \bibinfo{year}{2007}\natexlab{}.
\newblock \showarticletitle{Smooth sensitivity and sampling in private data analysis}. In \bibinfo{booktitle}{\emph{Proceedings of the Thirty-Ninth Annual ACM Symposium on Theory of Computing}} (San Diego, California, USA) \emph{(\bibinfo{series}{STOC '07})}. \bibinfo{publisher}{Association for Computing Machinery}, \bibinfo{address}{New York, NY, USA}, \bibinfo{pages}{75–84}.
\newblock
\showISBNx{9781595936318}
\urldef\tempurl%
\url{https://doi.org/10.1145/1250790.1250803}
\showDOI{\tempurl}


\bibitem[\protect\citeauthoryear{PAGE}{PAGE}{1954}]%
        {page1954continuous}
\bibfield{author}{\bibinfo{person}{E.~S. PAGE}.} \bibinfo{year}{1954}\natexlab{}.
\newblock \showarticletitle{CONTINUOUS INSPECTION SCHEMES}.
\newblock \bibinfo{journal}{\emph{Biometrika}} \bibinfo{volume}{41}, \bibinfo{number}{1-2} (\bibinfo{date}{06} \bibinfo{year}{1954}), \bibinfo{pages}{100--115}.
\newblock
\showISSN{0006-3444}
\urldef\tempurl%
\url{https://doi.org/10.1093/biomet/41.1-2.100}
\showDOI{\tempurl}
\showeprint{https://academic.oup.com/biomet/article-pdf/41/1-2/100/1243987/41-1-2-100.pdf}


\bibitem[\protect\citeauthoryear{Papadimitriou and Tsitsiklis}{Papadimitriou and Tsitsiklis}{1987a}]%
        {papadimitriou1987complexity}
\bibfield{author}{\bibinfo{person}{Christos~H. Papadimitriou} {and} \bibinfo{person}{John~N. Tsitsiklis}.} \bibinfo{year}{1987}\natexlab{a}.
\newblock \showarticletitle{The Complexity of Markov Decision Processes}.
\newblock \bibinfo{journal}{\emph{Math. Oper. Res.}} \bibinfo{volume}{12}, \bibinfo{number}{3} (\bibinfo{date}{Aug.} \bibinfo{year}{1987}), \bibinfo{pages}{441–450}.
\newblock
\showISSN{0364-765X}


\bibitem[\protect\citeauthoryear{Papadimitriou and Tsitsiklis}{Papadimitriou and Tsitsiklis}{1987b}]%
        {toutanova2015observed}
\bibfield{author}{\bibinfo{person}{Christos~H. Papadimitriou} {and} \bibinfo{person}{John~N. Tsitsiklis}.} \bibinfo{year}{1987}\natexlab{b}.
\newblock \showarticletitle{The Complexity of Markov Decision Processes}.
\newblock \bibinfo{journal}{\emph{Math. Oper. Res.}} \bibinfo{volume}{12}, \bibinfo{number}{3} (\bibinfo{date}{Aug.} \bibinfo{year}{1987}), \bibinfo{pages}{441–450}.
\newblock
\showISSN{0364-765X}


\bibitem[\protect\citeauthoryear{Patel, Kraft, Guestrin, and Zaharia}{Patel et~al\mbox{.}}{2024}]%
        {patel2024acorn}
\bibfield{author}{\bibinfo{person}{Liana Patel}, \bibinfo{person}{Peter Kraft}, \bibinfo{person}{Carlos Guestrin}, {and} \bibinfo{person}{Matei Zaharia}.} \bibinfo{year}{2024}\natexlab{}.
\newblock \showarticletitle{ACORN: Performant and Predicate-Agnostic Search Over Vector Embeddings and Structured Data}.
\newblock \bibinfo{journal}{\emph{Proc. ACM Manag. Data}} \bibinfo{volume}{2}, \bibinfo{number}{3}, Article \bibinfo{articleno}{120} (\bibinfo{date}{May} \bibinfo{year}{2024}), \bibinfo{numpages}{27}~pages.
\newblock
\urldef\tempurl%
\url{https://doi.org/10.1145/3654923}
\showDOI{\tempurl}


\bibitem[\protect\citeauthoryear{Rogers, Roth, Ullman, and Vadhan}{Rogers et~al\mbox{.}}{2016}]%
        {rogers2016privacy}
\bibfield{author}{\bibinfo{person}{Ryan Rogers}, \bibinfo{person}{Aaron Roth}, \bibinfo{person}{Jonathan Ullman}, {and} \bibinfo{person}{Salil Vadhan}.} \bibinfo{year}{2016}\natexlab{}.
\newblock \showarticletitle{Privacy odometers and filters: pay-as-you-go composition}. In \bibinfo{booktitle}{\emph{Proceedings of the 30th International Conference on Neural Information Processing Systems}} (Barcelona, Spain) \emph{(\bibinfo{series}{NIPS'16})}. \bibinfo{publisher}{Curran Associates Inc.}, \bibinfo{address}{Red Hook, NY, USA}, \bibinfo{pages}{1929–1937}.
\newblock
\showISBNx{9781510838819}


\bibitem[\protect\citeauthoryear{Sehgal and Salihoglu}{Sehgal and Salihoglu}{2025}]%
        {sehgal2025navix}
\bibfield{author}{\bibinfo{person}{Gaurav Sehgal} {and} \bibinfo{person}{Semih Salihoglu}.} \bibinfo{year}{2025}\natexlab{}.
\newblock \bibinfo{title}{NaviX: A Native Vector Index Design for Graph DBMSs With Robust Predicate-Agnostic Search Performance}.
\newblock
\newblock
\showeprint[arxiv]{2506.23397}~[cs.IR]
\urldef\tempurl%
\url{https://arxiv.org/abs/2506.23397}
\showURL{%
\tempurl}


\bibitem[\protect\citeauthoryear{Shapley}{Shapley}{1953}]%
        {shapley1953value}
\bibfield{author}{\bibinfo{person}{Lloyd~S. Shapley}.} \bibinfo{year}{1953}\natexlab{}.
\newblock \bibinfo{booktitle}{\emph{A Value for n-Person Games}}.
\newblock \bibinfo{publisher}{Princeton University Press}, \bibinfo{address}{Princeton, NJ}. 307--318 pages.
\newblock
\urldef\tempurl%
\url{https://doi.org/10.1515/9781400881970-018}
\showDOI{\tempurl}


\bibitem[\protect\citeauthoryear{Singh, Subramanya, Krishnaswamy, and Simhadri}{Singh et~al\mbox{.}}{2021}]%
        {singh2021freshdiskann}
\bibfield{author}{\bibinfo{person}{Aditi Singh}, \bibinfo{person}{Suhas~Jayaram Subramanya}, \bibinfo{person}{Ravishankar Krishnaswamy}, {and} \bibinfo{person}{Harsha~Vardhan Simhadri}.} \bibinfo{year}{2021}\natexlab{}.
\newblock \bibinfo{title}{FreshDiskANN: A Fast and Accurate Graph-Based ANN Index for Streaming Similarity Search}.
\newblock
\newblock
\showeprint[arxiv]{2105.09613}~[cs.IR]
\urldef\tempurl%
\url{https://arxiv.org/abs/2105.09613}
\showURL{%
\tempurl}


\bibitem[\protect\citeauthoryear{Tossou and Dimitrakakis}{Tossou and Dimitrakakis}{2015}]%
        {tossou2016algorithms}
\bibfield{author}{\bibinfo{person}{Aristide Tossou} {and} \bibinfo{person}{Christos Dimitrakakis}.} \bibinfo{year}{2015}\natexlab{}.
\newblock \bibinfo{title}{Algorithms for Differentially Private Multi-Armed Bandits}.
\newblock
\newblock
\showeprint[arxiv]{1511.08681}~[stat.ML]
\urldef\tempurl%
\url{https://arxiv.org/abs/1511.08681}
\showURL{%
\tempurl}


\bibitem[\protect\citeauthoryear{Wang and Jia}{Wang and Jia}{2023}]%
        {wang2023data}
\bibfield{author}{\bibinfo{person}{Jiachen~T. Wang} {and} \bibinfo{person}{Ruoxi Jia}.} \bibinfo{year}{2023}\natexlab{}.
\newblock \bibinfo{title}{Data Banzhaf: A Robust Data Valuation Framework for Machine Learning}.
\newblock
\newblock
\showeprint[arxiv]{2205.15466}~[cs.LG]
\urldef\tempurl%
\url{https://arxiv.org/abs/2205.15466}
\showURL{%
\tempurl}


\bibitem[\protect\citeauthoryear{Wu and Zhang}{Wu and Zhang}{2024}]%
        {wu2024faster}
\bibfield{author}{\bibinfo{person}{Hao Wu} {and} \bibinfo{person}{Hanwen Zhang}.} \bibinfo{year}{2024}\natexlab{}.
\newblock \showarticletitle{Faster differentially private top-k selection: a joint exponential mechanism with pruning}. In \bibinfo{booktitle}{\emph{Proceedings of the 38th International Conference on Neural Information Processing Systems}} (Vancouver, BC, Canada) \emph{(\bibinfo{series}{NIPS '24})}. \bibinfo{publisher}{Curran Associates Inc.}, \bibinfo{address}{Red Hook, NY, USA}, Article \bibinfo{articleno}{2266}, \bibinfo{numpages}{27}~pages.
\newblock
\showISBNx{9798331314385}


\bibitem[\protect\citeauthoryear{Wu, Jia, Lin, Huang, and Chang}{Wu et~al\mbox{.}}{2023a}]%
        {wu2023variance}
\bibfield{author}{\bibinfo{person}{Mengmeng Wu}, \bibinfo{person}{Ruoxi Jia}, \bibinfo{person}{Changle Lin}, \bibinfo{person}{Wei Huang}, {and} \bibinfo{person}{Xiangyu Chang}.} \bibinfo{year}{2023}\natexlab{a}.
\newblock \showarticletitle{Variance reduced Shapley value estimation for trustworthy data valuation}.
\newblock \bibinfo{journal}{\emph{Comput. Oper. Res.}} \bibinfo{volume}{159}, \bibinfo{number}{C} (\bibinfo{date}{Nov.} \bibinfo{year}{2023}), 9.
\newblock
\showISSN{0305-0548}
\urldef\tempurl%
\url{https://doi.org/10.1016/j.cor.2023.106305}
\showDOI{\tempurl}


\bibitem[\protect\citeauthoryear{Wu, Zhou, Tao, and Wang}{Wu et~al\mbox{.}}{2023b}]%
        {wu2023private}
\bibfield{author}{\bibinfo{person}{Yulian Wu}, \bibinfo{person}{Xingyu Zhou}, \bibinfo{person}{Youming Tao}, {and} \bibinfo{person}{Di Wang}.} \bibinfo{year}{2023}\natexlab{b}.
\newblock \showarticletitle{On private and robust bandits}. In \bibinfo{booktitle}{\emph{Proceedings of the 37th International Conference on Neural Information Processing Systems}} (New Orleans, LA, USA) \emph{(\bibinfo{series}{NIPS '23})}. \bibinfo{publisher}{Curran Associates Inc.}, \bibinfo{address}{Red Hook, NY, USA}, Article \bibinfo{articleno}{1511}, \bibinfo{numpages}{13}~pages.
\newblock


\bibitem[\protect\citeauthoryear{Xiao, Zhan, Xi, Hou, and Liao}{Xiao et~al\mbox{.}}{2024}]%
        {xiao2024enhancing}
\bibfield{author}{\bibinfo{person}{Wentao Xiao}, \bibinfo{person}{Yueyang Zhan}, \bibinfo{person}{Rui Xi}, \bibinfo{person}{Mengshu Hou}, {and} \bibinfo{person}{Jianming Liao}.} \bibinfo{year}{2024}\natexlab{}.
\newblock \bibinfo{title}{Enhancing HNSW Index for Real-Time Updates: Addressing Unreachable Points and Performance Degradation}.
\newblock
\newblock
\showeprint[arxiv]{2407.07871}~[cs.IR]
\urldef\tempurl%
\url{https://arxiv.org/abs/2407.07871}
\showURL{%
\tempurl}


\bibitem[\protect\citeauthoryear{Xu, Ruan, Korpeoglu, Kumar, and Achan}{Xu et~al\mbox{.}}{2020}]%
        {xu2020inductive}
\bibfield{author}{\bibinfo{person}{Da Xu}, \bibinfo{person}{Chuanwei Ruan}, \bibinfo{person}{Evren Korpeoglu}, \bibinfo{person}{Sushant Kumar}, {and} \bibinfo{person}{Kannan Achan}.} \bibinfo{year}{2020}\natexlab{}.
\newblock \bibinfo{title}{Inductive Representation Learning on Temporal Graphs}.
\newblock
\newblock
\showeprint[arxiv]{2002.07962}~[cs.LG]
\urldef\tempurl%
\url{https://arxiv.org/abs/2002.07962}
\showURL{%
\tempurl}


\bibitem[\protect\citeauthoryear{Xu, Liang, Li, Xu, Chen, Zhang, Li, Yang, Yang, Yang, Cheng, and Yang}{Xu et~al\mbox{.}}{2023}]%
        {xu2023spfresh}
\bibfield{author}{\bibinfo{person}{Yuming Xu}, \bibinfo{person}{Hengyu Liang}, \bibinfo{person}{Jin Li}, \bibinfo{person}{Shuotao Xu}, \bibinfo{person}{Qi Chen}, \bibinfo{person}{Qianxi Zhang}, \bibinfo{person}{Cheng Li}, \bibinfo{person}{Ziyue Yang}, \bibinfo{person}{Fan Yang}, \bibinfo{person}{Yuqing Yang}, \bibinfo{person}{Peng Cheng}, {and} \bibinfo{person}{Mao Yang}.} \bibinfo{year}{2023}\natexlab{}.
\newblock \showarticletitle{SPFresh: Incremental In-Place Update for Billion-Scale Vector Search}. In \bibinfo{booktitle}{\emph{Proceedings of the 29th Symposium on Operating Systems Principles}} (Koblenz, Germany) \emph{(\bibinfo{series}{SOSP '23})}. \bibinfo{publisher}{Association for Computing Machinery}, \bibinfo{address}{New York, NY, USA}, \bibinfo{pages}{545–561}.
\newblock
\showISBNx{9798400702297}
\urldef\tempurl%
\url{https://doi.org/10.1145/3600006.3613166}
\showDOI{\tempurl}


\bibitem[\protect\citeauthoryear{Yang, Ying, Shi, and Xing}{Yang et~al\mbox{.}}{2024}]%
        {lacroix2020tensor}
\bibfield{author}{\bibinfo{person}{Jinfa Yang}, \bibinfo{person}{Xianghua Ying}, \bibinfo{person}{Yongjie Shi}, {and} \bibinfo{person}{Bowei Xing}.} \bibinfo{year}{2024}\natexlab{}.
\newblock \showarticletitle{Tensor decompositions for temporal knowledge graph completion with time perspective▪}.
\newblock \bibinfo{journal}{\emph{Expert Syst. Appl.}} \bibinfo{volume}{237}, \bibinfo{number}{PA} (\bibinfo{date}{March} \bibinfo{year}{2024}), 12.
\newblock
\showISSN{0957-4174}
\urldef\tempurl%
\url{https://doi.org/10.1016/j.eswa.2023.121267}
\showDOI{\tempurl}


\bibitem[\protect\citeauthoryear{{Yelp}}{{Yelp}}{2026}]%
        {yelp2024dataset}
\bibfield{author}{\bibinfo{person}{{Yelp}}.} \bibinfo{year}{2026}\natexlab{}.
\newblock \bibinfo{title}{Yelp Open Dataset}.
\newblock
\newblock
\urldef\tempurl%
\url{https://www.yelp.com/dataset}
\showURL{%
\tempurl}


\bibitem[\protect\citeauthoryear{Yuan, Zhang, Du, Chen, Sun, Gao, Backes, He, and Chen}{Yuan et~al\mbox{.}}{2025}]%
        {yuan2024psgraph}
\bibfield{author}{\bibinfo{person}{Quan Yuan}, \bibinfo{person}{Zhikun Zhang}, \bibinfo{person}{Linkang Du}, \bibinfo{person}{Min Chen}, \bibinfo{person}{Mingyang Sun}, \bibinfo{person}{Yunjun Gao}, \bibinfo{person}{Michael Backes}, \bibinfo{person}{Shibo He}, {and} \bibinfo{person}{Jiming Chen}.} \bibinfo{year}{2025}\natexlab{}.
\newblock \bibinfo{title}{PSGraph: Differentially Private Streaming Graph Synthesis by Considering Temporal Dynamics}.
\newblock
\newblock
\showeprint[arxiv]{2412.11369}~[cs.CR]
\urldef\tempurl%
\url{https://arxiv.org/abs/2412.11369}
\showURL{%
\tempurl}


\bibitem[\protect\citeauthoryear{Zhang, Lee, Ma, Lou, Yang, and Xiong}{Zhang et~al\mbox{.}}{2024}]%
        {du2024dpar}
\bibfield{author}{\bibinfo{person}{Qiuchen Zhang}, \bibinfo{person}{Hong~kyu Lee}, \bibinfo{person}{Jing Ma}, \bibinfo{person}{Jian Lou}, \bibinfo{person}{Carl Yang}, {and} \bibinfo{person}{Li Xiong}.} \bibinfo{year}{2024}\natexlab{}.
\newblock \showarticletitle{DPAR: Decoupled Graph Neural Networks with Node-Level Differential Privacy}. In \bibinfo{booktitle}{\emph{Proceedings of the ACM Web Conference 2024}} (Singapore, Singapore) \emph{(\bibinfo{series}{WWW '24})}. \bibinfo{publisher}{Association for Computing Machinery}, \bibinfo{address}{New York, NY, USA}, \bibinfo{pages}{1170–1181}.
\newblock
\showISBNx{9798400701719}
\urldef\tempurl%
\url{https://doi.org/10.1145/3589334.3645531}
\showDOI{\tempurl}


\bibitem[\protect\citeauthoryear{Zhang, Wei, Engels, and Shun}{Zhang et~al\mbox{.}}{2025}]%
        {chen2023clean}
\bibfield{author}{\bibinfo{person}{Ziyu Zhang}, \bibinfo{person}{Yuanhao Wei}, \bibinfo{person}{Joshua Engels}, {and} \bibinfo{person}{Julian Shun}.} \bibinfo{year}{2025}\natexlab{}.
\newblock \bibinfo{title}{CleANN: Efficient Full Dynamism in Graph-based Approximate Nearest Neighbor Search}.
\newblock
\newblock
\showeprint[arxiv]{2507.19802}~[cs.DB]
\urldef\tempurl%
\url{https://arxiv.org/abs/2507.19802}
\showURL{%
\tempurl}


\bibitem[\protect\citeauthoryear{Zhong, Li, Jin, Yang, Chu, Wang, Shen, Jia, Gu, Xie, Lin, Shen, Song, and Cheng}{Zhong et~al\mbox{.}}{2025}]%
        {vsag2024}
\bibfield{author}{\bibinfo{person}{Xiaoyao Zhong}, \bibinfo{person}{Haotian Li}, \bibinfo{person}{Jiabao Jin}, \bibinfo{person}{Mingyu Yang}, \bibinfo{person}{Deming Chu}, \bibinfo{person}{Xiangyu Wang}, \bibinfo{person}{Zhitao Shen}, \bibinfo{person}{Wei Jia}, \bibinfo{person}{George Gu}, \bibinfo{person}{Yi Xie}, \bibinfo{person}{Xuemin Lin}, \bibinfo{person}{Heng~Tao Shen}, \bibinfo{person}{Jingkuan Song}, {and} \bibinfo{person}{Peng Cheng}.} \bibinfo{year}{2025}\natexlab{}.
\newblock \bibinfo{title}{VSAG: An Optimized Search Framework for Graph-Based Approximate Nearest Neighbor Search}.
\newblock , \bibinfo{numpages}{14}~pages.
\newblock
\showISSN{2150-8097}
\urldef\tempurl%
\url{https://doi.org/10.14778/3750601.3750624}
\showDOI{\tempurl}


\end{thebibliography}

\end{document}